\documentclass{article}
\usepackage[utf8]{inputenc}
\usepackage{array}
\usepackage{authblk,amsthm,amsmath}
\usepackage[natbib,backend=bibtex]{biblatex}
\usepackage{dmbiblatex}
\bibliography{cache_analysis_complexity}

\newcommand{\mytitle}{On the complexity of cache analysis for different replacement policies}
\title{\mytitle}

\author{David Monniaux}
\author{Valentin Touzeau}
\affil{Univ. Grenoble Alpes, CNRS, Grenoble INP\footnote{Institute of Engineering Univ. Grenoble Alpes}, VERIMAG, 38000 Grenoble, France}

\theoremstyle{plain}
\newtheorem{theorem}{Theorem}
\newtheorem{lemma}[theorem]{Lemma}
\newtheorem{corollary}[theorem]{Corollary}
\newtheorem{proposition}[theorem]{Proposition}

\theoremstyle{definition}
\newtheorem{definition}[theorem]{Definition}

\theoremstyle{remark}
\newtheorem{remark}[theorem]{Remark}
\newtheorem{example}[theorem]{Example}

\usepackage{listofitems,ifthen,subfigure}

\usepackage{tikz}
\usetikzlibrary{shapes.geometric}

\newcommand{\nWays}{N}
\newcommand{\nRegs}{r}
\newcommand{\true}{\mathbf{t}}
\newcommand{\false}{\mathbf{f}}

\newcommand{\vneg}[1]{\bar{#1}}

\newcommand{\plrueight}[4] {
	\tikzset{way/.style={draw, rectangle, minimum size=1cm, text height=1.5ex, text depth=.25ex}}
	\tikzset{tree/.style={draw, circle, minimum size=0.8cm, text height=1.5ex, text depth=.25ex}}
	\tikzset{arrow/.style={->, >=stealth}}

	\setsepchar{,}
	\readlist\ways{#3}

	\node[way] (w0) at (#1+0,#2+0) {\ways[1]};
	\node[way] (w1) at (#1+1,#2+0) {\ways[2]};
	\node[way] (w2) at (#1+2,#2+0) {\ways[3]};
	\node[way] (w3) at (#1+3,#2+0) {\ways[4]};
	\node[way] (w4) at (#1+4,#2+0) {\ways[5]};
	\node[way] (w5) at (#1+5,#2+0) {\ways[6]};
	\node[way] (w6) at (#1+6,#2+0) {\ways[7]};
	\node[way] (w7) at (#1+7,#2+0) {\ways[8]};

	\node[tree] (p0) at (#1+3.5,#2+3.2) {\StrChar{#4}{1}};

	\node[tree] (p1) at (#1+1.5,#2+2.2) {\StrChar{#4}{2}};
	\node[tree] (p2) at (#1+5.5,#2+2.2) {\StrChar{#4}{3}};

	\node[tree] (p3) at (#1+0.5,#2+1.2) {\StrChar{#4}{4}};
	\node[tree] (p4) at (#1+2.5,#2+1.2) {\StrChar{#4}{5}};
	\node[tree] (p5) at (#1+4.5,#2+1.2) {\StrChar{#4}{6}};
	\node[tree] (p6) at (#1+6.5,#2+1.2) {\StrChar{#4}{7}};

	\StrChar{#4}{1}[\bita]
	\StrChar{#4}{2}[\bitb]
	\StrChar{#4}{3}[\bitc]
	\StrChar{#4}{4}[\bitd]
	\StrChar{#4}{5}[\bite]
	\StrChar{#4}{6}[\bitf]
	\StrChar{#4}{7}[\bitg]

	\ifthenelse{\bita=0}
	           {\draw[arrow] (p0)--node[above left]{0}(p1);\draw        (p0)--node[above right]{1}(p2);}
	           {\draw        (p0)--node[above left]{1}(p1);\draw[arrow] (p0)--node[above right]{0}(p2);}
	\ifthenelse{\bitb=0}
	           {\draw[arrow] (p1)--node[above left]{0}(p3);\draw        (p1)--node[above right]{1}(p4);}
	           {\draw        (p1)--node[above left]{1}(p3);\draw[arrow] (p1)--node[above right]{0}(p4);}
	\ifthenelse{\bitc=0}
	           {\draw[arrow] (p2)--node[above left]{0}(p5);\draw        (p2)--node[above right]{1}(p6);}
	           {\draw        (p2)--node[above left]{1}(p5);\draw[arrow] (p2)--node[above right]{0}(p6);}
	\ifthenelse{\bitd=0}
	           {\draw[arrow] (p3)--node[left]{0}(w0);\draw        (p3)--node[right]{1}(w1);}
	           {\draw        (p3)--node[left]{1}(w0);\draw[arrow] (p3)--node[right]{0}(w1);}
	\ifthenelse{\bite=0}
	           {\draw[arrow] (p4)--node[left]{0}(w2);\draw        (p4)--node[right]{1}(w3);}
	           {\draw        (p4)--node[left]{1}(w2);\draw[arrow] (p4)--node[right]{0}(w3);}
	\ifthenelse{\bitf=0}
	           {\draw[arrow] (p5)--node[left]{0}(w4);\draw        (p5)--node[right]{1}(w5);}
	           {\draw        (p5)--node[left]{1}(w4);\draw[arrow] (p5)--node[right]{0}(w5);}
	\ifthenelse{\bitg=0}
	           {\draw[arrow] (p6)--node[left]{0}(w6);\draw        (p6)--node[right]{1}(w7);}
	           {\draw        (p6)--node[left]{1}(w6);\draw[arrow] (p6)--node[right]{0}(w7);}
}

\usepackage[pdftitle={\mytitle},pdfauthor={David Monniaux and Valentin Touzeau}]{hyperref}

\begin{document}
\maketitle

\begin{abstract}
  Modern processors use cache memory: a memory access that ``hits'' the cache returns early, while a ``miss'' takes more time.
Given a memory access in a program, cache analysis consists in deciding whether this access is always a hit, always a miss, or is a hit or a miss depending on execution.
Such an analysis is of high importance for bounding the worst-case execution time of safety-critical real-time programs.

There exist multiple possible policies for evicting old data from the cache when new data are brought in, and different policies, though apparently similar in goals and performance, may be very different from the analysis point of view.
In this paper, we explore these differences from a complexity-theoretical point of view. Specifically, we show that, among the common replacement policies, LRU (Least Recently Used) is the only one whose analysis is NP-complete, whereas the analysis problems for the other policies are PSPACE-complete.

\end{abstract}

\section{Introduction}
Most high performance processors implement some form of \emph{caching}: frequently used instructions and data are retained in fast memory close to the processing unit, to avoid costly requests from the main memory.
While the intuition is that a cache retains the most recently accessed memory words, up to its size, reality is far more complex: what happens depends on the number of cache levels, the size of each level, the ``number of ways'' (also known as the \emph{associativity}) of the cache and the \emph{cache replacement policy}, that is, the algorithm used for choosing which memory block to evict from the cache to make room for a new block.

Reading a memory block not in cache can take $10$ times to $100$ times the time needed to access it if it is cache.
Thus, static analysis approaches for bounding the worst-case execution time (WCET) of programs have to take into account whether or not data are cached.
Such analyses are used, for instance, for proving that the execution time of software in a critical control loop (e.g. in avionics) can never exceed the period of the loop.
Not only does caching, or lack of caching, directly influence execution time, it also complicates analysis itself, as different microarchitectural execution paths may be taken inside the processor depending on whether or not data are cached, and analysis has to take all these paths into account.

For these reasons, all static analyses for bounding execution time include a cache analysis, which determines which of the memory accesses made by the program are hits, which are misses, and which cannot be classified.
Analyses depend on the cache replacement policy, and, in the literature, there is a clear preference for the LRU (\emph{Least Recently Used}) policy, from the well-known age-based abstract analysis of Ferdinand~\cite{Ferdinand99} to recent work proved to be optimally precise in a certain sense~\cite{Touzeau_et_al_CAV2017,Touzeau_et_al_POPL19}.

In contrast, other policies such as PLRU (pseudo-LRU), NMRU and FIFO (\emph{First-In, First-Out}) have a reputation for being very hard to analyze~\cite{DBLP:journals/pieee/HeckmannLTW03} and for having poor predictability~\cite{DBLP:journals/rts/ReinekeGBW07}.
A legitimate question is whether these problems are intrinsically difficult, or is it just that research has not so far yielded efficient analyses.

Issues of static analysis of programs under different cache policies are not necessarily correlated with the practical efficiency of cache policies.
Static analysis is concerned with worst-case behavior,
and policies with approximately equal ``average''%
\footnote{By ``average'' we do not imply any probabilistic distribution, but rather an informal meaning over industrially relevant workloads, as opposed to examples concocted for exhibiting very good or very bad behavior.}
practical performance may be very different from the analysis point of view.
Even though PLRU and NMRU were designed as ``cheap'' (easier to implement in hardware) alternatives to LRU and have comparable practical efficiency~\cite{Al-Zoubi:2004:PEC:986537.986601}, they are very different from the  worst-case analysis point of view.
\medskip

In this paper, we explore these questions as decision problems:
\begin{definition}[Exist-Hit]
\label{def:exist-hit}
The \emph{exist-hit} problem is, for a given replacement policy:

\begin{center}
\begin{tabular}{>{\itshape}lp{10cm}}
Inputs & a control flow graph $G=(V, E)$ with edges adorned with memory block names\\
& a starting node $S$ in the graph\\
& a final node $F$ in the graph\\
& the cache \emph{associativity} (\emph{number of ways} $\nWays$), in unary\\
& a memory block name $a$\\
Outputs & a Boolean: is there an execution trace from $S$ to $F$, starting with an empty initial cache and ending with a cache containing $a$?
\end{tabular}
\end{center}
\end{definition}

\begin{definition}[Exist-Miss]
\label{def:exist-miss}
The \emph{exist-miss} problem is defined as above but with an ending state \emph{not} containing~$a$.
\end{definition}

We shall also study the variant of this problem where the initial cache contents are arbitrary:
\begin{definition}
The exist-hit (respectively, exist-miss) problem with arbitrary initial state contents is defined as above, except that the output is ``are there a legal initial cache state $\sigma$ and an execution trace from $S$ to $F$, starting in $\sigma$ and ending with a cache containing (respectively, not containing)~$a$?''.
\end{definition}

We shall here prove that
\begin{itemize}
\item for policies LRU, FIFO, pseudo-RR, PLRU, and NMRU, the exist-hit and exist-miss problems are NP-complete for acyclic control flow graphs;
\item for LRU, these problems are still NP-complete for cyclic control flow graphs;
\item for PLRU, FIFO, pseudo-RR, PLRU, and NMRU, these problems are PSPACE-complete for cyclic control flow graphs
\item for LRU, FIFO, pseudo-RR, and PLRU, the above results extend to exist-miss and exist-hit problems from an arbitrary starting state
\end{itemize}
Under the usual conjecture that PSPACE-complete problems are not in NP, this may justify why analyzing properties of FIFO, PLRU and NMRU caches is harder than for LRU.
\medskip

Real-life CPU cache systems are generally complex (multiple levels of caches) and poorly documented (often, the only information about replacement policies is by reverse engineering).
For our complexity-theoretical analyses we need simple models with clear mathematical definitions;
thus we consider only one level of cache, and only one ``cache set'' per cache.%
\footnote{%
A real cache system is composed of a large number of ``cache sets'': a memory block may fit in only one cache set depending on its address, and the replacement policy applies only within a given cache set.
For all commonly found cache replacement policies except pseudo-round-robin, and disregarding complex CPU pipelines, this means that the cache sets operate completely independently, each seeing only memory blocks that map to it;
each can be analyzed independently.
It is therefore very natural to consider the complexity of analysis over one single cache set, as we do in this paper.}

In this paper, we consider that the control-flow graph carries only identifiers of memory blocks to be accessed, abstracting away the data that are read or written, as well as arithmetic operations and guards.
Therefore, we take into account executions that cannot take place on the real system.
This is the same setting used by many static analyses for cache properties.
Some more precise static analyses attempt to discard some infeasible executions --- e.g. an execution with guards $x < 0$ and $x > 0$ with no intervening write to $x$ is infeasible.
In general, however, this entails deciding the reachability of program locations, a problem that is undecidable if the program operates over unbounded integers, and already PSPACE-complete if the program operates on a finite vector of bits~\cite{DBLP:conf/stacs/FeigenbaumKVV98};.
Clearly we cannot use such a setting to isolate the contribution of the cache analysis itself.


\section{Fixed associativity}
\label{sec:fixed_associativity}
In a given hardware cache, the associativity is fixed, typically $\nWays=2$, $4$, $8$, $12$ or $16$.
It thus makes sense to study cache analysis complexity for fixed associativity.
However, such analysis can always be done by explicit-state model-checking (enumeration of reachable states) in polynomial time:

\begin{theorem}
  Let us assume here that the associativity $\nWays$ is fixed, as well as the replacement policy (among those cited in this article).
  Then exist-hit and exist-miss properties can be checked in polynomial time, more precisely in $O(|G|^{\nWays+1})$ where $|G|$ is the size of the control-flow graph.
\end{theorem}

\begin{proof}
  Let $(V,E)$ be the control-flow graph; its size is $|G|=|V|+|E|$.
  Let $B$ the set of possible cache blocks.
  Without loss of generality, for all policies discussed in this article, the only blocks that matter in $B$ are those that are initially in the cache (at most $\nWays$) and those that are found on the control edges.
  Let us call the set of those blocks $B'$; $|B'| \leq |E| + \nWays$.

The state of the cache then consists in $\nWays$ blocks chosen among $|B'|$ possible ones, plus possibly some additional information that depends on the replacement policy (e.g. the indication that a line is empty); say $b$ bits per way.
The number of possible cache states is thus $(2^b|B'|)^{\nWays}$.

Let us now consider the finite automaton whose states are pairs $(p,\sigma)$ where $p$ is a node in the control-flow graph and $\sigma$ is the cache state, with the transition relation $(p,\sigma) \rightarrow (p',\sigma')$ meaning that the processor moves in one step from control node $p$ with cache state $\sigma$ to control node $p'$ with cache state $\sigma'$.
The number of states of this automaton is $|V|.(2^b|B'|)^{\nWays}$,
which is bounded by $|G|.(|G|+\nWays)^{\nWays}.2^{b\nWays}$, that is, $O(|G|^{\nWays+1})$.

Exist-miss and exist-hit properties amount to checking that certain states are reachable in this automaton.
This can be achieved by enumerating all reachable states of the automaton, which can be done in linear time in the size of the automaton. 
\end{proof}

It is an open question whether it is possible to find algorithms that are provably substantially better in the worst-case than this brute-force enumeration.
Also, would it be possible to separate replacement policies according to their growth with respect to associativity?
It is however unlikely that strong results of the kind ``PLRU analysis needs at least $K.|G|^{\nWays}$ operations in the worst case'' will appear soon, because they imply $\mathrm{P} \neq \mathrm{NP}$ or $\mathrm{P} \neq \mathrm{PSPACE}$.

\begin{theorem}
  Consider a policy among PLRU, FIFO, pseudo-RR (with known or unknown initial state) or NMRU with known initial state (respectively, LRU),
  and a problem among exist-miss and exist-hit.
  Assume $(H)$: for this policy, for any algorithm $A$ that decides this problem on this policy, and any associativity $\nWays$, there exist $K(\nWays)$ and $e(\nWays)$ such that for all $g_0$ there exists $g(\nWays,g_0) \geq g_0$ such that the worst-case complexity of $A$ on graphs of size $g$ is at least $K(\nWays).g(\nWays,g_0)^{e(\nWays)}$.
  Assume also $e(\nWays) \rightarrow \infty$ as $\nWays \rightarrow \infty$, then P is strictly included in PSPACE (respectively, NP).
\end{theorem}

\begin{proof}
  Suppose $(H')$: there exists a polynomial-time algorithm $A$ solving the analysis problem for arbitrary associativity, meaning that there exist a constant $K'$ and an exponent $e'$ such that $A'$ takes time at most $K'.(\nWays+g)^{e'}$ on a graph of size $g$ for associativity~$\nWays$.

  Let $\nWays$ be an associativity.
  From $(H)$ there is a strictly ascending sequence $g_m$ such that the worst-case complexity of $A$ on graphs of size $g_m$ is at least $K(\nWays).g_m^{e(\nWays)}$.
  From $(H')$, $K(\nWays).g_m^{e(\nWays)} \leq K'.(\nWays+g_m)^{e'}$.
  When $g_m \rightarrow \infty$ this is possible only if $e(\nWays) \leq e'$.

  Since $e(\nWays) \rightarrow \infty$ as $\nWays \rightarrow \infty$, the above is absurd.
  Thus there is no polynomial-time algorithm $A$ for solving the analysis problem for the given policy.
  We prove later in this paper that these analysis problems are PSPACE-complete for PLRU, FIFO, NMRU, pseudo-RR, and NP-complete for LRU; the result follows.
\end{proof}

\section{LRU}
\label{sec:LRU}
The ``Least Recently Used'' (LRU) replacement policy is simple and intuitive: the data block least recently used is evicted when a cache miss occurs.
The cache is thus a queue ordered by age: on a miss, the oldest block is discarded to make room for a new one, which has age 0; the ages of all other blocks are incremented.
If a block is already in the cache, it is ``rejuvenated'': its age is set to zero, and the ages of the blocks before it in the queue are incremented.

In other words,
the state of an LRU cache with associativity $\nWays$ is a word of length at most $\nWays$ over the alphabet of cache blocks, composed of pairwise distinct letters; an empty cache is defined by the empty word.
When an access is made to a block $a$, if it belongs to the word (\emph{hit}),
then this letter is removed from the word and appended to the word.
If it does not belong to the word (\emph{miss}), and the length of the word is less than $\nWays$, then $a$ is appended to the word;
otherwise, the length of the word is exactly $\nWays$ ---
the first letter of the word is discarded and $a$ is appended.

LRU has been used in Intel Pentium I, MIPS 24K/34K~\cite[p.21]{Reineke_PhD}, among others.
Notably, for Kalray processors K1a and K1b, LRU caches are advertised as advancing ``timing predictability''.

In this section, we shall extend our recent NP-hardness results~\cite{Touzeau_et_al_POPL19} to NP-completeness, and also prove NP-hardness for exist-hit on a restricted class of control-flow graphs.

\subsection{Motivation and fundamental properties}
LRU caches are appreciated by designers of static analysis tools that bound the worst-case execution time of the program, since an analysis based on abstract interpretation by~\citet{Ferdinand99} (basically, an interval for the age of each possible block) has long been known.
The analysis classifies each access in the program as ``always hit'' (all execution traces leading to that access produce a hit there), ``always miss'' (all execution traces leading to that access produce a miss there), or ``unknown''.
When it answers ``unknown'', it may be that it is in fact ``always hit'' or ``always miss'', but the analysis is too weak to come to a conclusion about it, or that there is at least one execution leading to a hit there and one leading to a miss there (``definitely unknown'').

In recent work~\cite{Touzeau_et_al_CAV2017,Touzeau_et_al_POPL19}, we closed that loophole and proposed an analysis that completely decides whether a given access is ``always hit'', ``always miss'', or ``definitely unknown'': the classical abstract interpretation is applied, along with another age-based abstract interpretation capable of concluding, in some cases, that an access is ``definitely unknown''; the remaining cases are decided by an exact but expensive (exponential worst-case) analysis.
These analyses solve both the exist-hit and the exist-miss problems;
was such an exponential cost unavoidable?
This motivated the studies in this paper.

Our analyses, as well as all our results on LRU in this paper, are based on the following easy, but fundamental, property of LRU caches:

\begin{proposition}\label{ref:LRU_access_count}
After an execution path starting from an empty cache, a block $a$ is in the cache if and only if there has been at least one access to $a$ along that path and the number of distinct blocks accessed since the last access to $a$ is at most $\nWays-1$.
\end{proposition}

\begin{example}
Assume a $4$-way cache, initially empty. After the sequence of accesses $bcabdcdb$, $a$ is in the cache because $bdcdb$ contains only $3$ distinct blocks $b,c,d$. In contrast, after the sequence $bcabdceb$, $a$ is no longer in the cache because $bdceb$ contains $4$ distinct blocks $b,c,d,e$.
\end{example}

\begin{remark}
  In definitions \ref{def:exist-hit} and \ref{def:exist-miss}, it does not matter if the associativity is specified in unary or binary. An associativity larger than the number of different blocks always produces hits, thus the problems become trivial.
  This also applies to FIFO, PLRU, NMRU caches and, more generally, to any cache analysis problem starting from an empty cache with a replacement policy that never evicts cache blocks as long as there is a free cache line. 
\end{remark} 

\subsection{Exist-Hit}
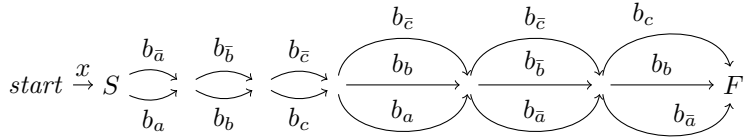
\begin{figure}
  \begin{center}
    \begin{tikzpicture}[->]
      \node (start) { \emph{start} };
      \node (q0) [right of=start] { $S$ };
      \node (q1) [right of=q0] { };
      \node (q2) [right of=q1] { };
      \node (q3) [right of=q2] { };

      \path (start) edge node[above] {$x$} (q0);
      \path (q0) edge[bend right] node[below] {$b_a$} (q1);
      \path (q0) edge[bend left] node[above] {$b_{\vneg{a}}$} (q1);

      \path (q1) edge[bend right] node[below] {$b_b$} (q2);
      \path (q1) edge[bend left] node[above] {$b_{\vneg{b}}$} (q2);

      \path (q2) edge[bend right] node[below] {$b_c$} (q3);
      \path (q2) edge[bend left] node[above] {$b_{\vneg{c}}$} (q3);

      \node (q4) [right of=q3, node distance=5em] { };
      \node (q5) [right of=q4, node distance=5em] { };
      \node (q6) [right of=q5, node distance=5em] { $F$ };

      \path (q3) edge[bend right=80] node[auto] {$b_a$} (q4);
      \path (q3) edge node[auto] {$b_b$} (q4);
      \path (q3) edge[bend left=80] node[auto] {$b_{\vneg{c}}$} (q4);

      \path (q4) edge[bend right=80] node[auto] {$b_{\vneg{a}}$} (q5);
      \path (q4) edge node[auto] {$b_{\vneg{b}}$} (q5);
      \path (q4) edge[bend left=80] node[auto] {$b_{\vneg{c}}$} (q5);

      \path (q5) edge[bend right=80] node[auto] {$b_{\vneg{a}}$} (q6);
      \path (q5) edge node[auto] {$b_b$} (q6);
      \path (q5) edge[bend left=80] node[auto] {$b_c$} (q6);
    \end{tikzpicture}
  \end{center}
  \vspace{-2em}

  \caption{There is a path from $q_0$ to $F$ with at most $3$ different labels if and only if the formula $(\vneg{c} \lor b \lor a) \land (\vneg{c} \lor \vneg{b} \lor \vneg{a}) \land (c \lor b \lor \vneg{a})$ has a model.
  Thus, for a LRU cache with associativity $\nWays=4$, there is an execution from $S$ to $F$ ending in a cache state containing $x$ if and only if this formula has a model.}
  \label{fig:exist-hit-reduction}
\end{figure}

\begin{theorem}\label{th:lru-exist-hit-acyclic-np-complete}
  The exist-hit problem is NP-complete for LRU and acyclic control-flow graphs.
\end{theorem}

\begin{proof}
  Obviously, the problem is in NP: a path may be chosen nondeterministically then checked in polynomial time.

  Now consider the following reduction from CNF-SAT (see \autoref{fig:exist-hit-reduction} for an example).
  Let $n_V$ be the number of variables in the SAT problem.
  With each variable $v$ in the SAT problem we associate two cache block labels $b_v$ and $b_{\vneg{v}}$.
  The idea is to represent a variable assignment as a cache state.
  We store $b_v$ in the cache when $v$ is set to true, and $b_{\vneg{v}}$ when $v$ is set to false.
  Let $x$ be a fresh name.
  We first load $x$ into the cache.
  The control-flow graph is then a sequence of switches:
  \begin{itemize}
  \item For each variable $v$ in the SAT problem, a switch between two edges labeled with $b_v$ and $b_{\vneg{v}}$ respectively.
  Executing this sequence of switches then loads into the cache a set of blocks representing a variable assignment.
  In addition, the next block that will be evicted (if any) is $x$.
  \item For each clause in the SAT problem, a switch between edges labeled with the blocks associated to the literals present in the clause.
  If at least one of these blocks has been accessed before, then the corresponding access can be performed without evicting $x$.
  Otherwise, $x$ is guaranteed to be evicted.
  \end{itemize}
  Each path through the sequence of switches with at most $n_V$ different labels corresponds to a SAT valid assignment, and conversely.
  Then there exists an execution such that at $F$ the cache contains $x$ if and only if there exists a SAT valid assignment.
\end{proof}

The objection can be made that the reduction in this proof produces control-flow graphs in which the same label occurs an arbitrary number of times --- the number of times the corresponding literal occurs in the CNF-SAT problem, plus one.
This is realistic for a data cache, since in a given program the same cache block may be accessed an arbitrary number of times.
It is however unrealistic for an instruction cache:%
\footnote{Unless procedure calls are ``inlined'' in the graph, because then the cache blocks corresponding to the inline procedures appear as many times as the number of locations it is called from.}
an instruction cache block has a fixed size, contains a maximum number of instructions, and thus cannot be accessed from an arbitrary number of control edges.%
\footnote{Consider a cache with 64-byte cache lines, as typical in x86 processors. In order for several basic blocks of instructions to overlap with that cache line, each, except perhaps the last one, must end with a branch instruction, which, in the shortest case, takes 2 bytes. No more than 32 basic blocks can overlap this cache line, and this upper bound is achieved by highly unrealistic programs.}
However, we can refine the preceding result to account for this criticism.

\begin{theorem}\label{th:lru-exist-hit-acyclic-np-complete-thrice}
  The exist-hit problem is NP-complete for LRU for acyclic control-flow graphs, even when the same cache block is accessed no more than thrice.
\end{theorem}

\begin{proof}
  We use the same reduction as in Th.~\ref{th:lru-exist-hit-acyclic-np-complete-thrice}, but from a CNF-SAT problem where each literal occurs at most twice, as per the following lemma.
\end{proof}

\begin{lemma}
  CNF-SAT is NP-hard even when restricted to sets of clauses where the same literal occurs at most twice, the same variable exactly thrice.%
  \footnote{We thank Pálvölgyi Dömötör for pointing out to us that this restriction is still NP-hard.}
\end{lemma}

\begin{proof}
  In the set of clauses, rename each occurrence of the same variable $v_i$ as a different variable name $v_{i,j}$, then add clauses $v_{i,1} \Rightarrow v_{i,2}$, $v_{i,2} \Rightarrow  v_{i,3}$, \dots, $v_{i,n-1} \Rightarrow v_{i,n}$, $v_{i,n} \Rightarrow v_{i,1}$ to establish logical equivalence between all renamings.
  Each literal now occurs once or twice, each variable thrice.
  Each model of the original formula corresponds to a model of the renamed formula, and conversely.
\end{proof}

\begin{remark}
The exist-hit problem is easy when the same cache block is accessed only once in the graph.
Assume that the aim is to test whether there exists an execution leading to a cache containing $x$ at the final node $F$.
Either there exists one reachable access $R$ to $x$ in the control-flow graph, or there is none (in the latter case, $x$ cannot be in the cache at node $F$)
.
Then there exists an execution leading to a cache state containing $x$ at node $F$ if and only if there exists a path of length at most $\nWays-1$ between $R$ and $F$ (see Proposition~\ref{ref:LRU_access_count}), which may be tested for instance by breadth-first traversal.
The complexity question remains opens when cache blocks are accessed at most twice.
\end{remark}

\begin{theorem}\label{th:lru-exist-hit-cyclic-np-complete}
  The exist-hit problem is still in NP for LRU when the graph may be cyclic.
\end{theorem}

\begin{proof}
  To prove that the problem is still in NP in case of cyclic graph, we show that the non-deterministic search of a exist-hit witness can be restricted to ``short'' paths (i.e.\ path of polynomial size).
  Consider a path $\pi$ from the starting node $S$ to the final node $F$ such that the final cache content contains $x$ (i.e.\ $\pi$ is a witness for the exist-hit problem).
  The idea of the proof is to remove accesses from $\pi$ to build a new witness $\pi'$ which can be found in polynomial time.

  The initial cache state being empty, there must be at least one access to $x$ in $\pi$.
  Let $i$ be the index of the last access to $x$:
  the edge $(\pi_i, \pi_{i+1})$ is labeled with block $x$, and for any $j > i, (\pi_j, \pi_{j+1})$ does not access $x$.
  We now split $\pi$ at index $i$ into two paths $\pi^1$ and $\pi^2$:
  $\pi^1 = \pi_1 \dots \pi_i$ and $\pi^2 = \pi_{i+1} \dots \pi_{|\pi|}$.
  $\pi^{1\prime}$ and $\pi^{2\prime}$ denote the paths obtained from $\pi^1$ and $\pi^2$ by removing cycles (subpaths beginning and ending in the same node).
  By construction, $\pi^{2\prime}$ does not evict $x$: the set of distinct memory blocks accessed along $\pi^{2\prime}$ is included in the set of distinct memory blocks accessed along $\pi^2$, which is insufficient to evict~$x$. 
  Thus, $x$ is guaranteed to be cached at the end of $\pi^{2\prime}$ (due to Property~\ref{ref:LRU_access_count}).
  Because $\pi^{1\prime}$ and $\pi^{2\prime}$ are cycle free, we have $|\pi^{1\prime}| \leq |V|$ and $|\pi^{2\prime}| \leq |V|$.
  The path $\pi' = \pi^{1\prime} \pi^{2\prime}$ obtained as the concatenation of $\pi^{1\prime}$ and $\pi^{2\prime}$ has thus length at most $2|V|$.
  In addition $\pi'$ goes from $S$ to $F$ and leads to a cache state containing $x$.
  Thus, the nondeterministic search for a witness hit path may be restricted to paths of length at most $2|V|$, which ensures membership in NP.
\end{proof}

\subsection{Exist-Miss}

\begin{figure}
\subfigure[Graph with (thick) Hamiltonian cycle]{
\hspace*{4em}
    \begin{tikzpicture}[node distance=3em]
      \node (q0) { $v'_0$ };
      \node (q1) [above right of=q0] { $v'_1$ };
      \node (q3) [below right of=q1] { $v'_3$ };
      \node (q2) [below right of=q0]  { $v'_2$ };
      \path (q0) edge[thick] (q1);
      \path (q0) edge[thick] (q2);
      \path (q1) edge (q2);
      \path (q1) edge[thick] (q3);
      \path (q2) edge[thick] (q3);
    \end{tikzpicture}
\hspace*{4em}
}
\hfill
\subfigure[Acyclic control-flow graph obtained by the reduction.
Edge labels are not shown;
the path corresponding to the Hamiltonian cycle is shown in a thick line.]{
\begin{tikzpicture}[->,node distance=3.5em,auto]
      \node (start) { $S$ };

      \node (q0s) [right of=start] { $v_0^0$ };

      \node (q2_1) [right of=q0s]  { $v_2^1$ };
      \node (q1_1) [above of=q2_1] { $v_1^1$ };
      \node (q3_1) [below of=q2_1] { $v_3^1$ };

      \node (q1_2) [right of=q1_1] { $v_1^2$ };
      \node (q2_2) [below of=q1_2]  { $v_2^2$ };
      \node (q3_2) [below of=q2_2] { $v_3^2$ };

      \node (q1_3) [right of=q1_2] { $v_1^3$ };
      \node (q2_3) [below of=q1_3]  { $v_2^3$ };
      \node (q3_3) [below of=q2_3] { $v_3^3$ };

      \node (q0e) [right of=q2_3] { $v_0^4$ };


      \path (q0s) edge[thick] (q1_1);
      \path (q0s) edge (q2_1);

      \path (q1_1) edge (q2_2);
      \path (q1_1) edge[thick] (q3_2);
      \path (q2_1) edge (q3_2);
      \path (q2_1) edge (q1_2);
      \path (q3_1) edge (q1_2);
      \path (q3_1) edge (q2_2);

      \path (q1_2) edge (q2_3);
      \path (q1_2) edge (q3_3);
      \path (q2_2) edge (q3_3);
      \path (q2_2) edge (q1_3);
      \path (q3_2) edge (q1_3);
      \path (q3_2) edge[thick] (q2_3);

      \path (q1_3) edge (q0e);
      \path (q2_3) edge[thick] (q0e);

      \path (start) edge[thick] node {$x$} (q0s);
\end{tikzpicture}}

\caption{Reduction from \autoref{th:lru-exist-miss-acyclic-np-complete} from the Hamiltonian cycle problem to the exist-miss problem for LRU caches.}
\label{fig:exist-miss-reduction}
\end{figure}
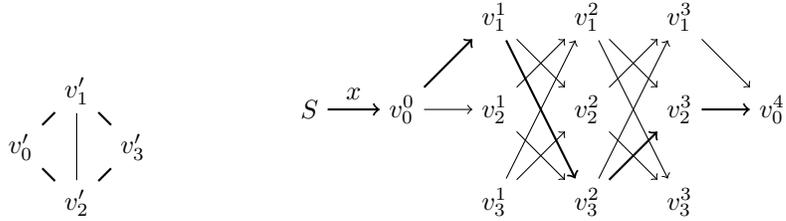

\begin{theorem}\label{th:lru-exist-miss-acyclic-np-complete}
  The exist-miss problem is NP-complete for LRU for acyclic control-flow graphs.
\end{theorem}

\begin{proof}
  Obviously, the problem is in NP: a path may be chosen nondeterministically, then checked in polynomial time.

  We reduce the Hamiltonian circuit problem to the exist-miss problem (see \autoref{fig:exist-miss-reduction} for an example).
  Let $G'=(V',E')$ be a graph, let $n = |V'|$, $V'=\{v'_0,\dots,v'_{n-1}\}$
  (the ordering is arbitrary).
  We check the existence of Hamiltonian circuit in $G'$.
  Let us construct an acyclic control-flow graph $G$ suitable for cache analysis as follows:
  \begin{itemize}
    \item two copies $v_0^0$ and $v_0^n$ of $v'_0$
    \item for each $v'_i$, $i \geq 1$, $|V'|-1$ copies $v_i^j$, $1 \leq j < n$
      (this arranges these nodes in layers indexed by~$j$)
    \item for each pair $v_i^j$, $v_{i'}^{j+1}$ of nodes in consecutive layers,
      an edge, labeled by the address~$i'$, if and only if there is
      an edge $(i,i')$ in~$E'$.
  \end{itemize}
  There is a Hamiltonian circuit in $G'$ if and only if there is a path in $G$ from $v_0^0$ to $v_0^n$ such that no edge label is repeated,
  thus if and only if there exists a path from $v_0^0$ to $v_0^n$ with at least $n$ distinct edge labels.
  Prepend an edge accessing a fresh block $x$ from the start node to $v_0^0$, then there exists a trace such that $x$ is not in the cache at $v_0^n$ if and only if this Hamiltonian circuit exists.
\end{proof}

The proof of \autoref{th:lru-exist-hit-cyclic-np-complete} (exist-hit problem is still in NP for cyclic CFG) does not carry over to the exist-miss case.
Indeed, this proof shows that if there is a path leading to a hit, then there is a ``short'' path that also lead to hit and can be discovered in polynomial time.
This ``short'' path might contain \emph{fewer} blocks between the last access to $x$ and the final node $F$ than the original witness.
Due to the fundamental property of LRU (see Property~\ref{ref:LRU_access_count}) this is not a problem and the ``short'' path is still guaranteed to lead to a hit.
However, in the case of the exist-miss problem, the short witness must be built carefully, as cutting cycles might remove accesses needed to evict $x$.
To show that the exist-miss problem for LRU is still in NP for cyclic control-flow graphs, one must exhibits short paths with the same blocks accessed (after $x$) than the original witness.
The following lemma ensures that this is always possible.

\begin{lemma}
  Let $B$ the set of memory blocks and $G=(V,E)$ a control-flow graph with edges decorated with elements of $B$ on edges.
  From any node $v_1$ and $v_2$ in $V$, and any path from $v_1 \in V$ to $v_2 \in V$ we can extract a path from $v_1$ to $v_2$ with the same contents (same memory blocks accessed) and length at most $|V|\cdot|B|$ (and thus at most $|V|\cdot|E|$).
\end{lemma}

\begin{proof}
  Consider a path $\pi$ from $v_1$ to $v_2$. $\pi$ can be segmented into sub-paths $\pi_1,\dots,\pi_m$, each beginning with the first occurrence of a new label not present in previous sub-paths.

  Each sub-path $\pi_i$ consists of an initial edge $e_i$ followed by $\pi'_i$. From $\pi'_i$ one can extract a simple path $\pi''_i$ --- that is, $\pi''_i$ has no repeated node --- of length at most $|V|-1$.
  By definition, there are no new edge label between $e_i$ and $e_{i+1}$.
  Thus removing cycles from $\pi'_i$ does not change the set of blocks accessed.
  Then, the concatenated path $e_1 \pi''_1 \cdots e_m \pi''_m$ has the same contents as $\pi$, starts and ends with the same nodes, and has at most $|V|\cdot|B|$ edges.
\end{proof}

\begin{theorem}\label{th:lru-exist-miss-cyclic-np-complete}
  The exist-miss problem is still in NP for cyclic control-flow graphs.
\end{theorem}

\begin{proof}
  Given a witness of miss existence $\pi$, one can split $\pi$ into $\pi^1$ and $\pi^2$ at the last access to $x$ (if any).
  If $\pi$ does not contain any access to $x$, then simply consider $\pi^1 = \varepsilon$ and $\pi^2 = \pi$.

  From the preceding lemma, one can extract from $\pi^1$ and $\pi^2$ two paths $\pi^{1\prime}$ and $\pi^{2\prime}$ of length at most $|V|\cdot|B|$ accessing the same sets of blocks.
  There are thus enough blocks accessed along $\pi^{2\prime}$ to evict $x$, and the path $\pi'$ obtained by chaining $\pi^{1\prime}$ and $\pi^{2\prime}$ is guaranteed to lead to a miss.
  One can thus search for a witness path of length at most $2\cdot|V|\cdot|E|$ to check the existence of a miss..
\end{proof}

\subsection{Extensions}
\begin{theorem}
The above theorems hold even if the starting cache state is unspecified: a problem with arbitrary starting cache state can be reduced to a problem of the same kind with empty starting cache state, and vice-versa.
\end{theorem}

\begin{proof}
  Consider a problem $P$ with empty initial cache state. Prepend to the control-flow graph of $P$ a sequence $f_1 \dots f_{\nWays}$ accesses to $\nWays$ pairwise distinct accesses to fresh blocks (blocks not appearing in $P$), where $\nWays$ is the associativity, and call the resulting problem $P'$.
  After executing this sequence, the cache only contains blocks from $f_1 \dots f_{\nWays}$ (not necessarily in that order) and none of the blocks of $P$. It is thus equivalent to check exist-hit or exist-miss properties on $P'$ from an arbitrary initial state and on $P$ from an empty initial state.

  Consider a problem $P$ with arbitrary initial cache state. Prepend to the control-flow graph of $P$ the following gadget, where $\epsilon$ denotes $\epsilon$-transitions (no memory access) and $b_1,\dots,b_m$ denote the alphabet of memory blocks:
    
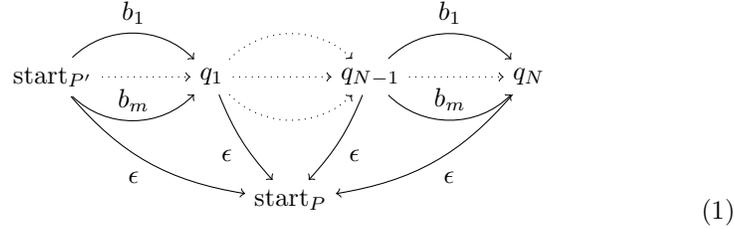
\begin{equation}
      \begin{tikzpicture}[node distance=6em]
        \node (q0) {$\text{start}_{P'}$};
        \node (q1) [right of=q0] {$q_1$};
        \node (qNm1) [right of=q1] {$q_{\nWays-1}$};
        \node (qN) [right of=qNm1] {$q_{\nWays}$};
        \path(q0) edge[->,bend left=45] node[above] {$b_1$} (q1);
        \path(q0) edge[->,dotted] (q1);
        \path(q0) edge[->,bend right=45] node[above] {$b_m$} (q1);
        \path(q1) edge[dotted,->,bend left=45] (qNm1);
        \path(q1) edge[->,dotted] (qNm1);
        \path(q1) edge[dotted,->,bend right=45] node(mid){} (qNm1);
        \path(qNm1) edge[->,bend left=45] node[above] {$b_1$} (qN);
        \path(qNm1) edge[->,dotted] (qN);
        \path(qNm1) edge[->,bend right=45] node[above] {$b_m$} (qN);
        \node (final) [below of=mid,node distance=3em] {$\text{start}_{P}$};
        \path(q0) edge[->,bend right=20] node[below left] {$\epsilon$} (final);
        \path(q1) edge[->,bend right=10] node[below left] {$\epsilon$} (final);
        \path(qNm1) edge[->,bend left=10] node[below right] {$\epsilon$} (final);
        \path(qN) edge[->,bend left=20] node[below right] {$\epsilon$} (final);
      \end{tikzpicture}
    \end{equation}
    This gadget loads into the cache any combination from zero to $\nWays$ blocks from $b_1,\dots,b_m$, in all possible orders.
    Thus analyzing $P'$ with a empty initial cache state is equivalent to analyzing $P$ with an arbitrary initial cache state.
\end{proof}

\begin{remark}
The proofs of NP-hardness for exist-hit and exist-miss on acyclic graphs for LRU carry over to FIFO (\autoref{sec:FIFO}).
\end{remark}


\section{Boolean register machine}
In the next sections, we shall prove that the exist-hit and the exist-miss problems for a variety of replacement policies are NP-hard for acyclic control-flow graphs and PSPACE-hard for general control-flow graphs.
All proofs will be by reduction from the reachability problem on a class of very simple machines, which we describe in this section: this problem is NP-complete if the control-flow graph of the machine is assumed to be acyclic, and PSPACE-complete in general.

\begin{definition}
A \emph{Boolean register machine} is defined by a number $\nRegs$ of registers and a directed (multi)graph with an initial node and a final node, with edges adorned by instructions of the form:
\begin{description}
\item[Guard] $v_i=b$ where $1 \leq i \leq \nRegs$ and $b \in \{\false, \true\}$,
\item[Assignment] $v_i:=b$ where $1 \leq i \leq \nRegs$ and $b \in \{\false, \true\}$.
\end{description}

The \emph{register state} is a vector of $\nRegs$ Booleans.
An edge with a guard $v_i=b$ may be taken only if the $i$-th register contains~$b$; the register state is unchanged.
The register state after the execution of an edge with an assignment $v_i:=b$ is the same as the preceding register state except that the $i$-th register now contains~$b$.

The \emph{reachability problem} for such a system is the existence of a valid execution starting in the initial node with all registers equal to $\false$, and leading to the final node.
\end{definition}

\begin{lemma}
The reachability problem for Boolean register machines is PSPACE-complete.
\end{lemma}

\begin{proof}
  Such a machine is easily simulated by a polynomial-space nondeterministic Turing machine; based on Savitch's theorem, the reachability problem is thus in PSPACE.

  Any Turing machine using space $P(|x|)$ on input $x$ can be simulated by a Boolean register machine with $O(P(|x|))$ registers, encoding the state of the tape of the Turing machine, and a number of transitions in $O(P(|x|).|D|)$ where $D$ is the description of the Turing machine.
\end{proof}

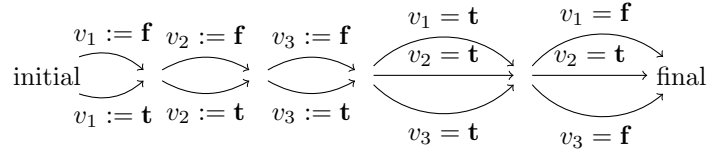
\begin{figure}
  \begin{center}
    \begin{tikzpicture}[node distance=4em,->]
      \node (q0) {initial};
      \node (q1) [right of=q0] {};
      \node (q2) [right of=q1] {};
      \node (q3) [right of=q2] {};
      \node (c1) [right of=q3, node distance=6em] {};
      \node (final) [right of=c1, node distance=6em] {final};
      \path (q0) edge [bend left] node[above] {$v_1:=\false$} (q1);
      \path (q0) edge [bend right] node[below] {$v_1:=\true$} (q1);
      \path (q1) edge [bend left] node[above] {$v_2:=\false$} (q2);
      \path (q1) edge [bend right] node[below] {$v_2:=\true$} (q2);
      \path (q2) edge [bend left] node[above] {$v_3:=\false$} (q3);
      \path (q2) edge [bend right] node[below] {$v_3:=\true$} (q3);
      \path (q3) edge [bend left=45] node[above] {$v_1=\true$} (c1);
      \path (q3) edge node[above] {$v_2=\true$} (c1);
      \path (q3) edge [bend right=45] node[below] {$v_3=\true$} (c1);
      \path (c1) edge [bend left=45] node[above] {$v_1=\false$} (final);
      \path (c1) edge node[above] {$v_2=\true$} (final);
      \path (c1) edge [bend right=45] node[below] {$v_3=\false$} (final);
    \end{tikzpicture}
    \caption{Reduction of CNF-SAT over $3$ unknowns with clauses $\{ v_1 \lor v_2 \lor v_3, \vneg{v_1} \lor v_2 \lor \vneg{v_3} \}$ to a Boolean $3$-register machine}
    \label{fig:Boolean_register_machine_NP}
  \end{center}
\end{figure}

\begin{lemma}
The reachability problem for Boolean register machines with acyclic control-flow is NP-complete.
\end{lemma}

\begin{proof}
A path from initial to final nodes, along with register values, may be guessed nondeterministically, then checked in polynomial time, thus reachability is in NP.

Any CNF-SAT problem with $\nRegs$ Boolean unknowns may be encoded as a Boolean $\nRegs$-register machine as follows: a sequence of $\nRegs$ disjunctions between $v_i:=\true$ and $v_i:=\false$ for each variable $i$, and then for each clause $v_{i^+_1} \lor \dots \lor v_{i^+_{n^+}} \lor
    v_{i^-_1} \lor \dots \lor v_{i^-_{n^-}}$, a disjunction between edges
$v_{i^+_i}=\true$ for all $1 \leq i \leq n^+$ and edges
$v_{i^-_i}=\false$ for all $1 \leq i \leq n^-$ (\autoref{fig:Boolean_register_machine_NP}).
\end{proof}

All forthcoming reductions will replace each instruction edge ($v_i=b$ guard edges, $v_i := b$ assignment edge) by a ``gadget'', a small acyclic piece of control-flow graph adorned with accesses to memory blocks;
they will also add a prologue and an epilogue.
The idea is to simulate executions of the Boolean register machine by executions of the cache system.

However, there is one problem: in Boolean register machines, an execution aborts when a guard is not satisfied,
whereas in our cache analysis problems, executions never abort except when reaching a control node with no outgoing edge.
We overcome this limitation by arranging the cache analysis problems so that the cache states following the simulation of an \emph{invalid} guard (e.g. the simulated state encodes $v_i = \true$ but the guard is $v_i = \false$) are irremediably marked as incorrectly formed;
thus the reachability problem for the Boolean register machine is reduced to the reachability of a correctly formed cache state at a certain node, which, in the epilogue, is encoded into the reachability of a cache hit (or a cache miss, respectively).

\section{FIFO}
\label{sec:FIFO}
FIFO (First-In, First-Out), also known as ``round-robin'', caches follow the same mechanism as LRU (a bounded queue ordered by age in the cache), except that a block is not rejuvenated on a hit.
They are used in Motorola PowerPC 56x, Intel XScale, ARM9, ARM11~\cite[p.21]{Reineke_PhD}, among others.

\subsection{Fundamental properties}
\label{sec:fifo_fundamental_properties}
The state of a FIFO cache with associativity $\nWays$ is a word of length at most $\nWays$ over the alphabet of cache blocks, composed of pairwise distinct blocks; an empty cache is defined by the empty word.
When an access is made to a block $a$, if it belongs to the word (\emph{hit}) then the cache state does not change.
If it does not belong to the word (\emph{miss}), and the length of the word is less than $\nWays$, then $a$ is appended to the word;
otherwise, the length of the word is exactly $\nWays$ ---
the first block of the word is discarded and $a$ is appended.

\begin{lemma}
  The exist-hit and the exist-miss problems are in NP for acyclic control flow graphs.
\end{lemma}

\begin{proof}
  Guess a path nondeterministically and execute the policy along it.
\end{proof}

\begin{lemma}
  The exist-hit and the exist-miss problems are in PSPACE for general graphs.
\end{lemma}

\begin{proof}
  Simulate the execution of the policy using a polynomial-space nondeterministic Turing machine. Based on Savitch's theorem, both problems are in PSPACE. 
\end{proof}

\subsection{Reduction to Exist-Hit}

We reduce the reachability problem for the Boolean register machine to the exist-hit problem for the FIFO cache as follows.
The associativity of the cache is chosen as $\nWays = 2\nRegs-1$.
The alphabet of cache blocks is
$\{ (a_{i,b})_{1 \leq i \leq \nRegs,b \in \{\false,\true\}} \} \cup
\{ (e_i)_{1 \leq i \leq \nRegs} \} \cup
\{ (f_i)_{1 \leq i \leq \nRegs-1} \} \cup
\{ (g_i)_{1 \leq i \leq \nRegs-1} \}$.

The main idea is to encode the value of registers by loading the blocks $a_{i,b}$ into the cache ($a_{i,\true}$ is used when the register $i$ contains value true, and $a_{i,\false}$ is used for false).
The blocks $e_{i}$ are used to distinguished valid boolean machine executions from executions where the machine should have halt.
Finally, blocks $f_i$ and $g_i$ are used in epilogue to turn valid states into cache hits and invalid states into cache misses.

The register state $v_1,\dots,v_{\nRegs}$ of the register machine is to be encoded as the FIFO state
\begin{equation}
  a_{1,v_1} e_2 a_{2,v_2} \dots e_{\nRegs} a_{\nRegs,v_{\nRegs}}.
\end{equation}
We use the FIFO state essentially as a delay-line memory.%
\footnote{We thank Ken McMillan for this remark.}

\begin{definition}
We say that a FIFO state of that form is \emph{well-formed at shift $1$}, or \emph{well-formed} for short if it is of the form
\begin{equation}
  a_{1,v_1} e_2 a_{2,v_2} \dots e_{\nRegs} a_{\nRegs,v_{\nRegs}}
\end{equation}

We say that a FIFO state is well-formed at shift $i$ ($2 \leq i \leq \nRegs$) if it is of the form
\begin{equation}
  a_{i,v_i} e_{i+1} a_{i+1,v_{i+1}} \dots a_{\nRegs,v_{\nRegs}} e_1
  a_{1,v_1} e_2 \dots a_{i-1,v_{i-1}}
\end{equation}

In both cases, we say that the FIFO state \emph{corresponds} to the state $v_1,\dots,v_{\nRegs}$.
This formalizes the notion of valid and invalid states.
\end{definition}

\begin{definition}\label{def:FIFO_reduction}
We turn the register machine graph into a cache analysis graph as follows.
\begin{itemize}
\item From the cache analysis initial node $I_f$ to the register machine former initial node $I_r$ there is a prologue, a sequence of accesses $a_{1,\false} e_2 \dots a_{\nRegs-1,\false} e_{\nRegs} a_{\nRegs,\false}$.

\item Each guard edge $v_i = b$ is replaced by the gadget
  
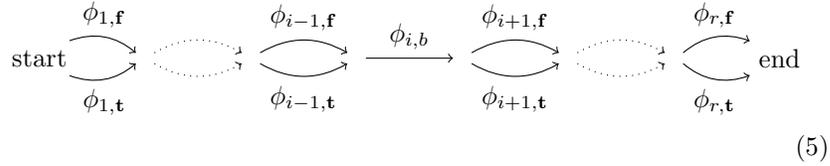
\begin{equation}
    \begin{tikzpicture}[node distance=4em,->]
      \node (q0) {start};
      \node (q1) [right of=q0] {};
      \node (qim2) [right of=q1] {};
      \node (qim1) [right of=qim2] {};
      \node (qi) [right of=qim1] {};
      \node (qip1) [right of=qi] {};
      \node (qRm1) [right of=qip1] {};
      \node (qR) [right of=qRm1] {end};
      \path (q0) edge [bend left] node[above] {$\phi_{1,\false}$} (q1);
      \path (q0) edge [bend right] node[below] {$\phi_{1,\true}$} (q1);
      \path (q1) edge [dotted, bend left] (qim2);
      \path (q1) edge [dotted, bend right] (qim2);
      \path (qim2) edge [bend left] node[above] {$\phi_{i-1,\false}$} (qim1);
      \path (qim2) edge [bend right] node[below] {$\phi_{i-1,\true}$} (qim1);
      \path (qim1) edge node[above] {$\phi_{i,b}$} (qi);
      \path (qi) edge [bend left] node[above] {$\phi_{i+1,\false}$} (qip1);
      \path (qi) edge [bend right] node[below] {$\phi_{i+1,\true}$} (qip1);
      \path (qip1) edge [dotted, bend left] (qRm1);
      \path (qip1) edge [dotted, bend right] (qRm1);
      \path (qRm1) edge [bend left] node[above] {$\phi_{\nRegs,\false}$} (qR);
      \path (qRm1) edge [bend right] node[below] {$\phi_{\nRegs,\true}$} (qR);
    \end{tikzpicture}
  \end{equation}
  where $\phi_{i,b}$ denotes the sequence of accesses
  $a_{i,b} e_i a_{i,b}$.

\item Each assignment edge $v_i := b$ is replaced by the gadget
  
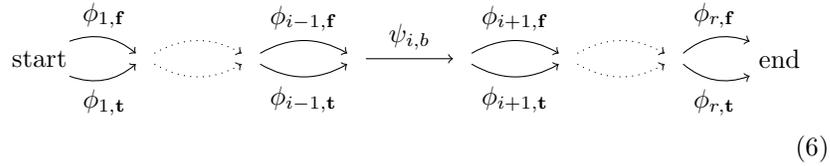
\begin{equation}
    \begin{tikzpicture}[node distance=4em,->]
      \node (q0) {start};
      \node (q1) [right of=q0] {};
      \node (qim2) [right of=q1] {};
      \node (qim1) [right of=qim2] {};
      \node (qi) [right of=qim1] {};
      \node (qip1) [right of=qi] {};
      \node (qRm1) [right of=qip1] {};
      \node (qR) [right of=qRm1] {end};
      \path (q0) edge [bend left] node[above] {$\phi_{1,\false}$} (q1);
      \path (q0) edge [bend right] node[below] {$\phi_{1,\true}$} (q1);
      \path (q1) edge [dotted, bend left] (qim2);
      \path (q1) edge [dotted, bend right] (qim2);
      \path (qim2) edge [bend left] node[above] {$\phi_{i-1,\false}$} (qim1);
      \path (qim2) edge [bend right] node[below] {$\phi_{i-1,\true}$} (qim1);
      \path (qim1) edge node[above] {$\psi_{i,b}$} (qi);
      \path (qi) edge [bend left] node[above] {$\phi_{i+1,\false}$} (qip1);
      \path (qi) edge [bend right] node[below] {$\phi_{i+1,\true}$} (qip1);
      \path (qip1) edge [dotted, bend left] (qRm1);
      \path (qip1) edge [dotted, bend right] (qRm1);
      \path (qRm1) edge [bend left] node[above] {$\phi_{\nRegs,\false}$} (qR);
      \path (qRm1) edge [bend right] node[below] {$\phi_{\nRegs,\true}$} (qR);
    \end{tikzpicture}
  \end{equation}
  where $\psi_{i,b}$ denotes the sequence of accesses $e_i a_{i,b} e_i$.
\item From the register machine former final node $F_r$ to a node $F_a$ there is a sequence of accesses
  $\psi_{1,\false} \dots \psi_{\nRegs,\false}$,
  constituting the first part of the epilogue.
\item From $F_a$ to a node $F_h$ there is a sequence of accesses 
  $a_{1,\false} g_1 e_2 f_2 a_{2,\false} g_2 \dots e_{\nRegs-1} f_{\nRegs-1} a_{\nRegs-1,\false} g_{\nRegs-1} e_{\nRegs} f_{\nRegs}$,
  constituting the second part of the epilogue.
\item The final node is $F_f = F_h$.
\end{itemize}
\end{definition}

The main difficulty in this reduction is that the Boolean register machines may terminate traces if a guard is not satisfied, whereas the cache problem has no guards and no way to terminate traces.
Our workaround is that cache states that do not correspond to traces from the Boolean machine are irremediably marked as incorrect (formally: \emph{well-phased but not well-formed}, per the following definition).

\begin{definition}
We say that a FIFO state of that form is \emph{well-phased at shift $1$}, or \emph{well-phased} for short if it is of the form
\begin{equation}
  \beta_1 \alpha_2 \beta_2 \dots \alpha_{\nRegs} \beta_{\nRegs}
\end{equation}
where, for each $i$:
\begin{itemize}
  \item either $\alpha_i = e_i$ and $\beta_i = a_{i,b_i}$ for some $b_i$, 
  \item or $\beta_i = e_i$ and $\alpha_i = a_{i,b_i}$ for some $b_i$.
\end{itemize}
    
We say that a FIFO state is well-phased at shift $i$ ($2 \leq i \leq \nRegs$) if it is of the form
\begin{equation}
  \beta_i \alpha_{i+1} \beta_{i+1} \dots \alpha_{\nRegs} \beta_{\nRegs}
  \alpha_1 \beta_1 \dots \alpha_{i-1} \beta_{i-1}
\end{equation}
\end{definition}

\begin{lemma}
  Assume $w$ is well-formed at shift $i$, corresponding to state $\sigma=(\sigma_1,\dots,\sigma_{\nRegs})$.
  If $\sigma_i = b$, then executing $\phi_{i,b}$ over FIFO state $w$ leads to a state well-formed at shift $i+1$ ($1$ if $i=\nRegs$), corresponding to $\sigma$~too.
  If $\sigma_i = \neg b$, then executing $\phi_{i,b}$ over FIFO state $w$ leads to a state well-phased, but not well-formed, at shift $i+1$ ($1$ if $i=\nRegs$).
\end{lemma}

\begin{proof}
  Without loss of generality we prove this for $i=1$ and $b=\false$.
  Assume $w= a_{1,\false} e_2 a_{2,v_2} \dots e_{\nRegs} a_{\nRegs,v_{\nRegs}}$; then the sequence $\phi_{1,\false} = a_{1,\false} e_1 a_{1,\false}$ yields $a_{2,v_2} \dots e_{\nRegs} a_{\nRegs,v_{\nRegs}} e_1 a_{1,\false}$.
  Assume now $w= a_{1,\true} e_2 a_{2,v_2} \dots e_{\nRegs} a_{\nRegs,v_{\nRegs}}$;
  then $\phi_{1,\false}$ yields $a_{2,v_2} \dots e_{\nRegs} a_{\nRegs,v_{\nRegs}} a_{1,\false} e_1$.
\end{proof}

\begin{lemma}
  Assume $w$ is well-formed at shift $i$, corresponding to state $\sigma=(\sigma_1,\dots,\sigma_{\nRegs})$.
  Executing $\psi_{i,b}$ over FIFO state $w$ leads to a state well-formed at shift $i+1$ ($1$ if $i=\nRegs$), corresponding to $\sigma$ where $\sigma_i$ has been replaced by~$b$.
\end{lemma}

\begin{proof}
  Without loss of generality we prove it for $i=1$ and $b=\false$.
  Assume $w= a_{1,v_1} e_2 a_{2,v_2} \dots e_{\nRegs} a_{\nRegs,v_{\nRegs}}$; then the sequence $\psi_{1,\false} = e_1 a_{1,\false} e_1$ yields $a_{2,v_2} \dots e_{\nRegs} a_{\nRegs,v_{\nRegs}} e_1 a_{1,\false}$.
\end{proof}

\begin{corollary}
  Assume starting in a well-formed FIFO state, corresponding to state $\sigma$, then any path through the gadget encoding an assignment or a guard
  \begin{itemize}
  \item either leads to a well-formed FIFO state, corresponding to the state $\sigma'$ obtained by executing the assignment, or $\sigma'=\sigma$ for a valid guard;
  \item or leads to a well-phased but not well-formed state.
  \end{itemize}
\end{corollary}

\begin{lemma}
  Assume $w$ is well-phased, but not well-formed, at shift $i$, then executing $\psi_{i,b}$ or  $\phi_{i,b}$ over FIFO state $w$ leads to a state well-phased, but not well-formed, at shift $i+1$.
\end{lemma}

\begin{proof}
  Without loss of generality, we shall prove this for $i=1$.
  Let $w = \beta_1 \alpha_2 \beta_2 \dots \alpha_{\nRegs} \beta_{\nRegs}$.
  
  First case: $\beta_1 = e_1$.
  $\psi_{1,b} = e_1 a_{1,b} e_1$ then leads to $\beta_2 \alpha_3 \beta_3 \dots \alpha_{\nRegs} \beta_{\nRegs} a_{1,b} e_1$, which is well-phased, but not well-formed due to the last two blocks, at shift~$2$.
  $\phi_{1,b} = a_{1,b} e_1 a_{1,b}$ also leads to the same state.

  Second case: $\beta_1$ is either $a_{1,\false}$ or $a_{1,\true}$; assume the former without loss of generality.
  Then there exists $j > 1$ such that
  $\alpha_j = a_{j, v_j}$ and $\beta_j = e_j$.
  $\psi_{1,b}$ then leads to $\beta_2 \alpha_3 \beta_3 \dots \alpha_{\nRegs} \beta_{\nRegs} e_1 a_{1,b}$, which is well-phased, but not well-formed due to the $\alpha_j,\beta_j$, at shift~$2$.

  $\phi_{1,\false}$ leads to $\beta_2 \alpha_3 \beta_3 \dots \alpha_{\nRegs} \beta_{\nRegs} e_1 a_{1,\false}$, which is well-phased, but not well-formed due to the $\alpha_j,\beta_j$, at shift~$2$.
  
  $\phi_{1,\true}$ leads to $\beta_2 \alpha_3 \beta_3 \dots \alpha_{\nRegs} \beta_{\nRegs} a_{1,\true} e_1$, which is well-phased, but not well-formed due to the last two blocks, at shift~$2$.
\end{proof}

\begin{corollary}
  Assume starting in a well-phased but not well-formed FIFO state, then any path through the gadget encoding an assignment or a guard leads to a well-phased but not well-formed FIFO state.
\end{corollary}

\begin{corollary}
  Any path from a well-formed FIFO state in $I_r$ to $F_r$ in the FIFO graph from \autoref{def:FIFO_reduction}
  \begin{itemize}
  \item either corresponds to a valid sequence of assignments and guards from the register machine from $I_r$ to $F_r$, and leads to a well-formed FIFO state corresponding to the final state of that sequence
  \item or corresponds to an invalid sequence of assignments and guards from the register machine, and leads to a well-phased but not well-formed FIFO state.
  \end{itemize}

  Conversely, any valid sequence of assignments and guards from the register machine maps from $I_r$ to $F_r$ transforms a well-formed FIFO state into a well-formed FIFO state, corresponding respectively to the initial and final states of that sequence.
\end{corollary}

\begin{corollary}
  The path from $F_r$ to $F_a$ (as in Definition~\ref{def:FIFO_reduction}):
  \begin{itemize}
  \item transforms a well-phased but not well-formed FIFO state into a well-phased but not well-formed FIFO state
  \item transforms any well-formed FIFO state into a well-formed FIFO state $w_0$ corresponding to the initial register state (all registers zero).
  \end{itemize}
\end{corollary}

\begin{lemma}
  The path $a_{1,\false} g_1 e_2 f_2 a_{2,\false} g_2 \dots e_{\nRegs-1} f_{\nRegs-1} a_{\nRegs-1,\false} g_{\nRegs-1} e_{\nRegs} f_{\nRegs}$ (from $F_a$ to $F_h$ in Definition~\ref{def:FIFO_reduction}):
  \begin{itemize}
  \item transforms $w_0$ into $a_{\nRegs,\false} g_1 f_2 g_2 \dots f_{\nRegs-1} g_{\nRegs-1} f_{\nRegs}$
  \item transforms any other word $w$ consisting of $a$'s and $e$'s into a word not containing~$a_{\nRegs,\false}$.
  \end{itemize}
\end{lemma}

\begin{proof}
  The first item is trivial. We shall now prove that it is necessary for the input word to be exactly $w_0$ in order for the final word to contain $a_{\nRegs,\false}$.
  In order for that, there must have been at most $2\nRegs-2$ misses along the path.
  The accesses to
  $g_1, f_2, g_2, \dots, f_{\nRegs-1}, g_{\nRegs-1}, f_{\nRegs}$ are always misses.
  As there are $2\nRegs-2$ of them, there must have been exactly those misses and no others.
  This implies that $a_{\nRegs,\false}$ was in the last position in $w$.
  
  When $e_{\nRegs}$ is processed, similarly there were exactly $2\nRegs-3$ misses, and $e_{\nRegs}$ must be a hit.
  This implies that $e_{\nRegs}$ was in the last or the penultimate position in $w$, but since the last position was occupied by  $a_{\nRegs,\false}$, $e_{\nRegs}$ must have been in the penultimate position.

  The same reasoning holds for all preceding locations, down to the first one, and thus the lemma holds.
\end{proof}

From all these lemmas, the main result follows:
\begin{corollary}
There is an execution of the FIFO cache from $I_f$ to $F_f$ such that $a_{\nRegs,\false}$ is in the final cache state if and only if there is an execution of the Boolean register machine from $I_r$ to $F_r$.
\end{corollary}

\begin{theorem}
  The exist-hit problem for FIFO caches is NP-complete for acyclic graphs and PSPACE-complete for general graphs.
\end{theorem}

\begin{proof}
  As seen above, a register machine reachability problem can be reduced in polynomial time to a exist-hit FIFO problem, preserving acyclicity if applicable.
\end{proof}

\begin{remark}
  We have described a reduction from a $\nRegs$-register machine to a FIFO cache problem with an odd $2\nRegs-1$ number of ways.
  This reduction may be altered to yield an even number of ways as follows.
  Two special padding blocks $p$ and $p'$ are added.
  A well-formed state is now $p a_{1,b_1} e_2 a_{2,b_2} \dots a_{\nRegs,b_{\nRegs}}$;
  the definition of well-phased states is similarly modified.
  Each gadget $G$ for assignment or guard is replaced by $p' G p G$.
  The first $p'$ turns padding $p$ into $p'$, $G$ is applied.
  The second $p'$ turns $p'$ into $p$ and $G$ is applied again.
  
  This remark also applies to the exist-miss problem.
\end{remark}

\subsection{Reduction to Exist-Miss}
We modify the reduction for exist-hit in order to exhibit a miss on $a_{r,\false}$ later on if and only if it is in the cache at the end of the graph defined above.

\begin{definition}\label{def:FIFO_reduction_miss}
  We transform the register machine graph into a cache analysis graph as in \autoref{def:FIFO_reduction}, with the following modification:
  in between $F_h$ and $F_f$ we insert a sequence
  $a_{\nRegs,\false} e_1 a_{1,\false} \dots e_{\nRegs-1} a_{\nRegs-1,\false} e_{\nRegs}$,
  constituting the third part of the epilogue.
\end{definition}

\begin{lemma}
  The path from $F_h$ to $F_f$ transforms
  $a_{\nRegs,\false} g_1 f_2 g_2 \dots f_{\nRegs-1} g_{\nRegs-1} f_{\nRegs}$ into
  a word not containing $a_{\nRegs,\false}$.
  It transforms any word composed of $f$'s and $g$'s only into a word containing~$a_{\nRegs,\false}$.
\end{lemma}

\begin{theorem}
  The exist-miss problem for FIFO caches is NP-complete for acyclic graphs and PSPACE-complete for general graphs.
\end{theorem}

\subsection{Extension to arbitrary starting cache}
\begin{lemma}\label{lemma:FIFO_arbitrary_start}
  The exist-hit and exist-miss problems for an empty starting FIFO cache state are reduced, in linear time, to the same kind of problem for an arbitrary starting cache state, with the same associativity.
\end{lemma}

\begin{proof}
  Let $\Sigma$ be the alphabet of blocks in the problem and $\nWays$ its associativity.
  Let $e_1,\dots,e_{2\nWays-1}$ be new blocks not in~$\Sigma$; after accessing them in sequence, the cache contains only elements from these accesses~\cite[Th.~1]{DBLP:journals/rts/ReinekeGBW07}.
  Prepend this sequence as a prologue to the cache problem; then the rest of the execution of the cache problem will behave as though it started from an empty cache.
\end{proof}

\begin{corollary}
  The exist-hit and exist-miss problems for FIFO caches with arbitrary starting state is NP-complete for acyclic graphs and PSPACE-complete for general graphs.
\end{corollary}

\subsection{Extension to Pseudo-RR caches}
Recall how a FIFO cache with multiple cache sets --- the usual approach in hardware caches --- operates.
A memory block of address $x$ is stored in the cache set number $H(x)$ where $H$ is a suitable function, normally a simple combination of the bits of~$x$.
In typical situations, this is as though the address $x$ were specified as a pair $(s,a)$ where $s$ is the number of the cache set and $a$ is the block name to be used by the FIFO in cache set number~$s$.

In a FIFO cache, each cache set, being a FIFO, can be implemented as a circular buffer: an array of cache blocks and a ``next to be evicted'' index.
In contrast, in a pseudo-RR cache, the ``next to be evicted'' index is global to all cache sets.

A FIFO cache exist-hit or exist-miss problem with cache blocks  $a_1,\dots,a_n$ can be turned into an equivalent pseudo-RR problem simply by using $(s,a_1),\dots,(s,a_n)$ as addresses for a constant distinguished cache set~$s$.
Thus, both exist-hit and exist-miss are NP-hard for acyclic control-flow graphs on pseudo-RR caches, and PSPACE-hard for general control-flow graphs.

The same simulation arguments used for FIFO (\autoref{sec:fifo_fundamental_properties}) hold for establishing membership in NP and PSPACE respectively.


\section{PLRU}
\label{sec:PLRU}
Because LRU caches were considered too difficult to implement efficiently in hardware, various schemes for heuristically approximating the behavior of a LRU cache (keeping the most recently used data) have been proposed.
By ``heuristically approximating'' we mean that these schemes are assumed, on ``typical'' workloads, to perform close to LRU, even though worst-case performance may be different.%
\footnote{%
  Experimentally, on typical workloads, the tree-based PLRU scheme described in this section is said to produce 5\% more misses on a level-1 data cache compared to LRU~\cite{Al-Zoubi:2004:PEC:986537.986601}.
  However, that scheme may, under specific concocted workloads, indefinitely keep data that are actually never used except once---
  a misperformance that cannot occur with LRU~\cite{DBLP:journals/pieee/HeckmannLTW03}.
  This can produce \emph{domino effects}: the cache behavior of a loop body may be indefinitely affected by the cache contents before the loop~\cite{berg:OASIcs:2006:672}.

Because of the difficulties in obtaining justifiable bounds on the worst-case execution times of programs running on a PLRU cache, some designers of safety-critical real-time systems lock all cache ways except for two, exploiting the fact that a 2-way PLRU cache is the same as a 2-way LRU cache and thus recovering predictability~\cite[\S3]{berg:OASIcs:2006:672}.}
Some authors lump all such schemes as ``pseudo-LRU'' or ``PLRU'', and call the scheme in the present section ``tree-based PLRU'' or ``PLRU-t''~\cite{Al-Zoubi:2004:PEC:986537.986601}, while some others~\cite[p.~26]{Reineke_PhD} call ``PLRU'' only the scheme discussed here.

PLRU has been used in i486~\cite{i486_data_sheet} (4-way), Intel Pentium II-IV, PowerPC 75x~\cite[p.~21]{Reineke_PhD};
an 8-way PLRU is used in NXP/Freescale MPC745x~\cite[p.~3-41]{MPC7450_Manual} and MPC75x~\cite[p.~3-19]{MPC750_Manual}, e6500, MPC8540.

\subsection{PLRU caches}
The cache lines of a PLRU cache, which may contain cached blocks, are arranged as the leaves of a full binary tree --- thus the number of ways $\nWays$ is a power of $2$, often $4$ or~$8$.
Two lines may not contain the same block.
Each internal node of the tree has a tag bit, which is represented as an arrow pointing to the left or right branch.
The state of the cache is thus the content of the lines and the $\nWays-1$ tag bits.

There is always a unique line such that there is a sequence of arrows from the root of the tree to the line; this is the line \emph{pointed at by the tags}.
Tags are said to be \emph{adjusted away} from a line as follows: on the path from the root of the tree to the line, tag bits are adjusted so that the arrows all point away from that path.

When a block $a$ is accessed:
\begin{itemize}
\item If the block is already in the cache, tags are adjusted away from this line.
\item If the block is not already in the cache and one or more cache lines are empty, the leftmost empty line is filled with $a$, and tags are adjusted away from this block.
\item If the block is not already in the cache and no cache line is empty, the block pointed at by the tags is evicted and replaced with $a$, and tags are adjusted away from this block.
\end{itemize}

\subsection{Exist-Hit Problem}
\label{sec:eh_plru}
We reduce the reachability problem of a Boolean $\nRegs$-register machine to the PLRU exist-hit problem for a $(2\nRegs+2)$-way cache --- without loss of generality, we can always add useless registers so that $2\nRegs+2$ is a power of two.
The alphabet of cache blocks is
$\{ (a_{i,b})_{1 \leq i \leq \nRegs,b \in \{\false,\true\}} \} \cup
 \{ (e_i)_{0 \leq i \leq \nRegs} \} \cup
 \{ c \}$.

\begin{definition}
We say that a PLRU cache state is \emph{well-formed} and \emph{corresponds} to a Boolean state $(b_i)_{1 \leq i \leq \nRegs}$ if its leaves are, from left to right:
$c, e_0, a_{1,b_1}, e_1, \dots, a_{\nRegs,b_{\nRegs}}, e_{\nRegs}$.
\end{definition}

\begin{definition}
We say that a PLRU cache state is \emph{well-phased} if its leaves are, from left to right: $x_0, e_0, a_{1,b_1}, e_1, \dots, a_{\nRegs,b_{\nRegs}}, e_{\nRegs}$ where $x_0$ can be $c$ or any $a_{i,b}$ block.
\end{definition}

We use the PLRU state as a random access memory. Appropriate sequence of accesses define the memory location to be read or written.

\begin{lemma}
Let $0 \leq i \leq \nRegs$, there exists a sequence $\pi_i$ of accesses, of length logarithmic in $\nRegs$, such that, when run on a well-phased cache state $x_0, e_0, \dots, x_{\nRegs}, e_{\nRegs}$, that sequence makes tags point at $x_i$ without changing the contents of the cache lines.
\end{lemma}

\begin{proof}
Moving from $x_i$ to the root of the tree, at every node one $e_i$ block is accessed from the other branch at that node (see Figure~\ref{fig:plru_pi_sequence} for an example of sequence $\pi_1$).
\end{proof}

\begin{figure}
  \begin{tikzpicture}[scale=0.7, every node/.style={scale=0.7}]
    \plrueight{0}{6}{$x_0$,$e_0$,$x_1$,$e_1$,$x_2$,$e_2$,$x_3$,$e_3$}{1101110}
    \plrueight{9}{6}{$x_0$,$e_0$,$x_1$,$e_1$,$x_2$,$e_2$,$x_3$,$e_3$}{1001010}
    \plrueight{9}{0}{$x_0$,$e_0$,$x_1$,$e_1$,$x_2$,$e_2$,$x_3$,$e_3$}{1100010}
    \plrueight{0}{0}{$x_0$,$e_0$,$x_1$,$e_1$,$x_2$,$e_2$,$x_3$,$e_3$}{0110000}

    \path[draw,thick,->] (7, 8) -- node[above]{$e_1$} (9, 8);
    \path[draw,thick,->] (12.5, 5) -- node[right]{$e_0$} (12.5, 4);
    \path[draw,thick,->] (9, 2) -- node[above]{$e_2$} (7, 2);
  \end{tikzpicture}
  \caption{Sequence $\pi_1 = e_1 e_0 e_2$ makes tags point at $x_1$ without changing cache content}
  \label{fig:plru_pi_sequence}
\end{figure}
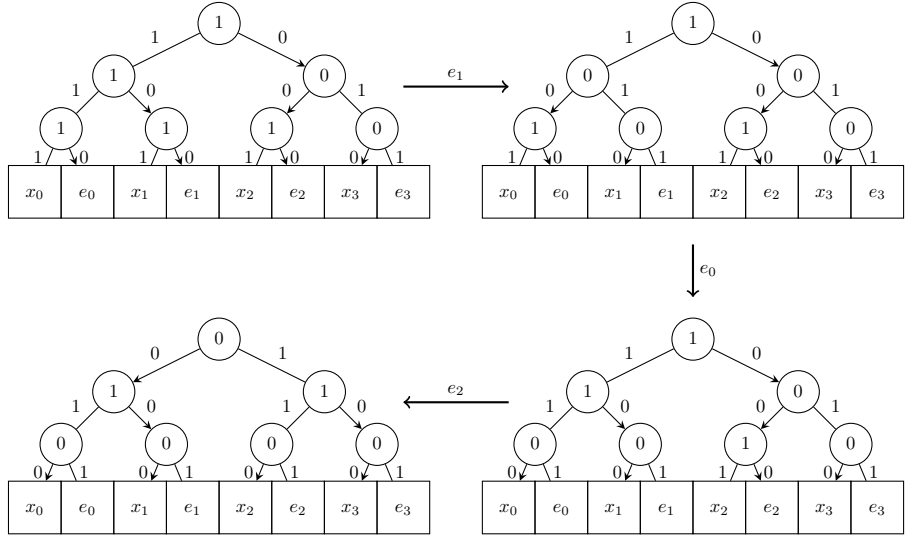

Let $1 \leq i \leq \nRegs$, $b \in \{ \true, \false \}$. Let $\phi_{i,b}$ be the sequence $\pi_0 a_{i,b}$, and $\psi_{i,b}$ the sequence $\pi_i a_{i,b}$.

\begin{definition}\label{def:PLRU_reduction}
We turn the register machine graph into a cache analysis graph as follows.
\begin{itemize}
\item From the cache analysis initial node $I_p$ to the register machine former initial node $I_r$ there is a sequence of accesses $c, e_0, a_{1, \false}, e_1, \dots, a_{\nRegs,\false}, e_{\nRegs}$.
\item Each guard edge $v_i = b$ is replaced by the sequence $\phi_{i,b}$, and each assignment edge  $v_i := b$ by the sequence $\psi_{i,b}$.
\item The cache final node $F_p$ is the same as the register machine final node~$F_r$.
\end{itemize}
\end{definition}

The idea is that any missed guard irremediably removes $c$ from the cache.
The following lemmas are easily proved by symbolically simulating the execution of the gadgets over the cache states:

\begin{lemma}
$\phi_{i,b}$ and $\psi_{i,b}$ map any well-phased but not well-formed state to a well-phased but not well-formed state.
\end{lemma}

\begin{lemma}
$\psi_{i,b}$ maps a well-formed state to a well-formed state corresponding to the same Boolean state where register $i$ has been replaced by~$b$.
\end{lemma}

\begin{lemma}
$\phi_{i,b}$ maps a well-formed state corresponding to a Boolean state $(b_i)_{1 \leq i \leq \nRegs}$ to
\begin{itemize}
\item if $b_i=b$, a well-formed state corresponding to the same Boolean state;
\item otherwise, a well-phased but not well-formed state.
\end{itemize}
\end{lemma}

\begin{corollary}
There is an execution of the PLRU cache from $I_p$ to $F_p$ such that $c$ is in the final cache state if and only if there is an execution of the Boolean register machine from $I_r$ to $F_r$.
\end{corollary}

\begin{proof}
A well-phased state is well-formed if and only if it contains~$c$.
\end{proof}

\begin{theorem}
  The exist-hit problem for PLRU caches is NP-complete for acyclic graphs and PSPACE-complete for general graphs.
\end{theorem}

\subsection{Exist-Miss Problem}
\label{sec:em_plru}

We use two extra blocks $d$ and $f$.

\begin{definition}
Let $Z$ be the sequence $\pi_{\nRegs} d \pi_{\nRegs} c \pi_0 f$.
\end{definition}

\begin{lemma}
$Z$ turns any well-formed state into a state not containing~$c$.
$Z$ turns any well-phased but not well-formed state into a state containing~$c$.
\end{lemma}

\begin{proof}
Consider a well-formed state $c e_0 a_{1,b_1} e_1 \dots a_{\nRegs,b_{\nRegs}} e_{\nRegs}$.
$Z$ replaces $a_{\nRegs,b_{\nRegs}}$ by $d$; then the access to $c$ does not change the line contents since $c$ is in the cache, and $c$ is replaced by~$f$.

Consider a well-phased but not well-formed state $x_0 e_0 a_{1,b_1} e_1 \dots a_{\nRegs,b_{\nRegs}} e_{\nRegs}$ where $x_0 \neq c$.
$Z$ replaces $a_{\nRegs,b_{\nRegs}}$ by $d$; then the access to $c$ replaces $d$ by $c$, and $x_0$ is replaced by~$f$.
\end{proof}

\begin{definition}\label{def:PLRU_reduction_EM}
We turn the register machine graph into a cache analysis graph in the same manner as in Definition~\ref{def:PLRU_reduction}, but between $F_r$ and $F_p$ we insert $Z$ as epilogue.
\end{definition}

\begin{lemma}
There is an execution of the PLRU cache from $I_p$ to $F_p$ such that $c$ is not in the final cache state if and only if there is an execution of the Boolean register machine from $I_r$ to $F_r$.
\end{lemma}

\begin{theorem}
  The exist-miss problem for PLRU caches is NP-complete for acyclic graphs and PSPACE-complete for general graphs.
\end{theorem}

\subsection{Extension to an arbitrary starting cache}
\citet[Th.~12]{DBLP:journals/rts/ReinekeGBW07} proved the following result on PLRU caches:

\begin{theorem}\label{th:PLRU_evict_all}
  It takes at most $\frac{k}{2} \log_2 k +1$  pairwise different accesses to evict all entries from a $k$-way set-associative PLRU cache set. Again, this is a tight bound.
\end{theorem}

More precisely, they prove that there is a sequence of accesses of length $\frac{k}{2} \log_2 k +1$ such that after executing this sequence over a cache with arbitrary initial state, the cache contains only elements from the sequence.

\begin{lemma}
  The exist-hit and exist-miss problems for an empty starting PLRU cache state are reduced, in linear time, to the same kind of problem for an arbitrary starting cache state, with the same associativity.
\end{lemma}

\begin{proof}
  Same proof as \autoref{lemma:FIFO_arbitrary_start}, except we need a sequence of $M=\frac{\nWays}{2} \log_2 \nWays + 1$ new blocks.

  Let $\Sigma$ be the alphabet of blocks in the problem and $\nWays$ its associativity.
  Let $e_1,\dots,e_M$ be new blocks not in~$\Sigma$; after accessing them in a sequence as constructed in Theorem~\ref{th:PLRU_evict_all}, there is no longer any block from $\Sigma$ in the cache.
  Prepend this sequence as a prologue to the cache problem; then the rest of the execution of the cache problem will behave as though it started from an empty cache.
\end{proof}

\begin{corollary}
  The exist-hit and exist-miss problems for FIFO caches with arbitrary starting state is NP-complete for acyclic graphs and PSPACE-complete for general graphs.
\end{corollary}


\section{NMRU}
\label{sec:NMRU}
Other forms of ``pseudo-LRU'' schemes have been proposed than the one discussed in \autoref{sec:PLRU}.
One of them, due to~\citet{MRU_patent} is based on the use of ``most recently used'' bits.
It is thus sometimes referred to as the ``not most recently used'' (NMRU) policy, or ``PLRU-m''~\cite{Al-Zoubi:2004:PEC:986537.986601}.
Confusingly, some literature~\cite{Reineke_PhD} also refers to this policy as ``MRU'' despite the fact that in this policy, it is \emph{not} the most recently used data block that is evicted first.

NMRU is used in the Intel Nehalem architecture, among others.

\subsection{NMRU caches}
\begin{definition}
The state of an $\nWays$-way NMRU cache is a sequence of at most $\nWays$ memory blocks $\alpha_i$, each tagged by a $0/1$ ``MRU-bit'' $r_i$ saying whether the associated block is to be considered not recently used ($0$) or recently used ($1$),
denoted by $\alpha_1^{r_1} \dots \alpha_{\nWays}^{r_\nWays}$.

An access to a block in the cache, a \emph{hit}, results in the associated MRU-bit being set to~$1$. If there were already $\nWays-1$ MRU-bits equal to $1$, then all the other MRU-bits are set to~$0$.

An access to a block $a$ not in the cache, a \emph{miss}, results in:
\begin{itemize}
\item if the cache is not full (number of blocks less than $\nWays$), then $a^1$ is appended to the sequence
\item if the cache is full (number of blocks equal to $\nWays$), then the leftmost (least index~$i$) block with associated MRU-bit $0$ is replaced by~$a^1$.
If there were already $\nWays-1$ MRU-bits equal to $1$, then all the other MRU-bits are set to~$0$.
\end{itemize}
\end{definition}

\begin{remark}
This definition is correct because the following invariant is maintained: either the cache is not full, or it is full but at least one MRU-bit is zero.
\end{remark}

\begin{example}
Assume $\nWays=4$. If the cache contains $a^0 b^0 c^0$, then an access to $d$ yields $a^0 b^0 c^0 d^1$ since the cache was not full. If $a$ is then accessed, the state becomes $a^1 b^0 c^0 d^1$. If $e$ is then accessed, the state becomes $a^1 e^1 c^0 d^1$ since $b$ was the leftmost block with a zero MRU-bit.
If $f$ is then accessed, then the state becomes $a^0 e^0 f^1 d^0$.
\end{example}

\subsection{Reduction to Exist-Hit}

We reduce the reachability problem for the register machine to the exist-hit problem for the NMRU cache as follows.
The associativity of the cache is chosen as $\nWays = 2\nRegs+3$.
The alphabet of the cache blocks is
$\{ (a_{i,b})_{1 \leq i \leq \nRegs,b \in \{\false,\true\}}\} \cup
\{ (e_i)_{1 \leq i \leq \nRegs} \} \cup
\{ (c_i)_{1 \leq i \leq \nRegs} \} \cup
\{ d \} \cup
\{ g_0, g_1 \}$.

The register state $v_1,\dots,v_{\nRegs}$ of the register machine is to be encoded as the NMRU state
\begin{equation}
  e_1^0 \dots e_\nRegs^0 d^0 a_{1,v_1} \dots a_{\nRegs,v_\nRegs}^0 g_0^0 g_1^1
\end{equation}
where the exponent (0 or 1) is the MRU-bit associated with the block.

\begin{definition}\label{def:NMRU_reduction}
We turn the register machine graph into a cache analysis graph as follows.
\begin{itemize}
\item From the cache analysis initial node $I_f$ to the register machine former initial node $I_r$ there is the prologue: the sequence of accesses $e_1 \dots  e_\nRegs d a_{1,\false} \dots a_{\nRegs,\false} g_0 g_1$.

\item Each guard edge $v_i = b$ is replaced by the gadget $\phi_{i,b} g_0 \phi_{i,b} g_1$, where $\phi_{i,b}$ is
  
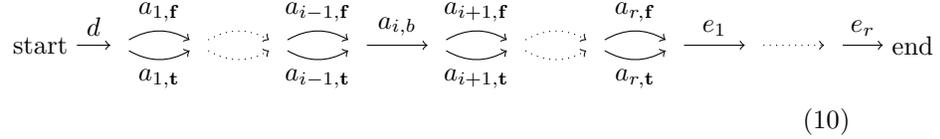
\begin{equation}
    \begin{tikzpicture}[node distance=3em,->]
      \node (q0) {start};
      \node (qd) [right of=q0] {};
      \node (q1) [right of=qd] {};
      \node (qim2) [right of=q1] {};
      \node (qim1) [right of=qim2] {};
      \node (qi) [right of=qim1] {};
      \node (qip1) [right of=qi] {};
      \node (qRm1) [right of=qip1] {};
      \node (qR) [right of=qRm1] {};
      \path (q0) edge node[above] {$d$} (qd);
      \path (qd) edge [bend left] node[above] {$a_{1,\false}$} (q1);
      \path (qd) edge [bend right] node[below] {$a_{1,\true}$} (q1);
      \path (q1) edge [dotted, bend left] (qim2);
      \path (q1) edge [dotted, bend right] (qim2);
      \path (qim2) edge [bend left] node[above] {$a_{i-1,\false}$} (qim1);
      \path (qim2) edge [bend right] node[below] {$a_{i-1,\true}$} (qim1);
      \path (qim1) edge node[above] {$a_{i,b}$} (qi);
      \path (qi) edge [bend left] node[above] {$a_{i+1,\false}$} (qip1);
      \path (qi) edge [bend right] node[below] {$a_{i+1,\true}$} (qip1);
      \path (qip1) edge [dotted, bend left] (qRm1);
      \path (qip1) edge [dotted, bend right] (qRm1);
      \path (qRm1) edge [bend left] node[above] {$a_{\nRegs,\false}$} (qR);
      \path (qRm1) edge [bend right] node[below] {$a_{\nRegs,\true}$} (qR);
      \node (qb1) [right of=qR] {};
      \node (qbRm1) [right of=qb1] {};
      \node (qbR) [right of=qbRm1] {end};
      \path (qR) edge node[above] {$e_1$} (qb1);
      \path (qb1) edge [dotted] (qbRm1);
      \path (qbRm1) edge node[above] {$e_\nRegs$} (qbR);
    \end{tikzpicture}
  \end{equation}

\item Each assignment edge $v_i := b$ is replaced by the gadget $\psi_{i,b} g_0 \psi_{i,b} g_1$, where $\psi_{i,b}$ is
  
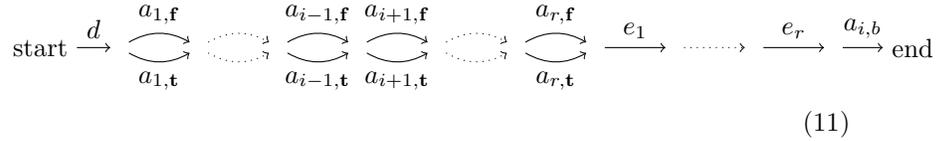
\begin{equation}
    \begin{tikzpicture}[node distance=3em,->]
      \node (q0) {start};
      \node (qd) [right of=q0] {};
      \node (q1) [right of=qd] {};
      \node (qim2) [right of=q1] {};
      \node (qim1) [right of=qim2] {};
      \node (qip1) [right of=qim1] {};
      \node (qRm1) [right of=qip1] {};
      \node (qR) [right of=qRm1] {};
      \path (q0) edge node[above] {$d$} (qd);
      \path (qd) edge [bend left] node[above] {$a_{1,\false}$} (q1);
      \path (qd) edge [bend right] node[below] {$a_{1,\true}$} (q1);
      \path (q1) edge [dotted, bend left] (qim2);
      \path (q1) edge [dotted, bend right] (qim2);
      \path (qim2) edge [bend left] node[above] {$a_{i-1,\false}$} (qim1);
      \path (qim2) edge [bend right] node[below] {$a_{i-1,\true}$} (qim1);
      \path (qim1) edge [bend left] node[above] {$a_{i+1,\false}$} (qip1);
      \path (qim1) edge [bend right] node[below] {$a_{i+1,\true}$} (qip1);
      \path (qip1) edge [dotted, bend left] (qRm1);
      \path (qip1) edge [dotted, bend right] (qRm1);
      \path (qRm1) edge [bend left] node[above] {$a_{\nRegs,\false}$} (qR);
      \path (qRm1) edge [bend right] node[below] {$a_{\nRegs,\true}$} (qR);
      \node (qb1) [right of=qR] {};
      \node (qbRm1) [right of=qb1] {};
      \node (qbR) [right of=qbRm1] {};
      \node (qi) [right of=qbR] {end};
      \path (qR) edge node[above] {$e_1$} (qb1);
      \path (qb1) edge [dotted] (qbRm1);
      \path (qbRm1) edge node[above] {$e_\nRegs$} (qbR);
      \path (qbR) edge node[above] {$a_{i,b}$} (qi);
    \end{tikzpicture}
  \end{equation}

\item From the register machine former final node $F_r$ to a node $F_a$ there is a sequence of gadgets for the assignments $v_1 := \false \dots v_r := \false$, the first part of the epilogue.
\item From $F_a$ to a node $F_h$ there is a sequence of accesses
  $a_{1,\false} \dots a_{\nRegs,\false} c_1 \dots c_\nRegs$, the second part of the epilogue.
\item The final node is $F_f = F_h$.
\end{itemize}
\end{definition}

\begin{definition}
We say that an NMRU state is well-formed  at step $s \in \{0,1\}$ if it is of the form
\begin{equation}
  \beta_1^0 \dots \beta_\nRegs^0 d^0 \alpha_1^0 \dots \alpha_\nRegs^0 g_0^s g_1^{1-s}
\end{equation}
where $\forall i, 1 \leq i \leq \nRegs, \alpha_i \in \{a_{\sigma(i),\false}, a_{\sigma(i),\true}\}, \beta_i = e_{\sigma'(i)}$ and $\sigma$ and $\sigma'$ are two permutations of $[1,\nRegs]$.
In other words, a well-formed state contains $\nRegs$ distinct blocks $e_i$ placed before $d$, and $\nRegs$ blocks $a_{i,b}$, with distinct $i$'s, placed between $d$ and $g_0$.
We say ``well-formed'' for short if~$s=0$.
\end{definition}

\begin{definition}
We say that an NMRU state is well-phased at step $s \in \{0,1\}$ if it is of the form
\begin{equation}
  \gamma_{\sigma(1)}^0 \dots \gamma_{\sigma(\nRegs)}^0 d^0 \gamma_{\sigma(\nRegs+1)}^0 \dots \gamma_{\sigma(2\nRegs)}^0 g_0^s g_1^{1-s}
\end{equation}
where $\gamma_1 = e_1, \dots, \gamma_\nRegs = e_\nRegs, \gamma_{\nRegs+1} \in \{a_{1,\false}, a_{1,\true}\}, \dots, \gamma_{2\nRegs} \in \{a_{\nRegs,\false}, a_{\nRegs,\true}\}$ and $\sigma$ is a permutation of $[1,2\nRegs]$.
We say ``well-phased'' for short if~$s=0$.
\end{definition}

\begin{lemma}
  Executing a path through $\phi_{i,b} g_s$ over an NMRU state $w$ well-phased at step $s$ always leads to a state well-phased at step $1-s$.
  Furthermore that state
  \begin{itemize}
  \item either is not well-formed at step $1-s$
  \item or is identical to $w$ except for the $g_0$ and $g_1$ blocks, and this may occur only if $a_{i,b}$ belongs to~$w$.
  \end{itemize}
\end{lemma}

\begin{proof}
  The input state $w$ is $x_1^0, \dots, x_r^0, d^0, x_{r+1}^0, \dots, x_{2r}^0, g_0^0, g_1^1$ where the $(x_i)_{1 \leq i \leq 2r}$ are a permutation of
  $\{ e_i \mid 1 \leq i \leq r \} \cup
  \{ a_{i,\beta_i} \mid 1 \leq i \leq r \}$ for some sequence of Booleans
  $(\beta_i)_{1 \leq i \leq r}$.

  Consider a path through $\phi_{i,b}$: it consists of $d$, followed by a sequence of $r$ $a$'s, then $r$ $e$'s.
  Each of these accesses either freshens, or overwrites, one of the $x$ positions.
  After the sequence of $a$'s, there are either no $a$'s to the left of $d$, or at least one.
  The former case is possible only if all $a$'s are hits, freshening positions to the right of $d$ --- this means all these positions are left untouched except that their MRU bits are flipped to $1$.
  Then the sequence of $e$'s just flips to $1$ the MRU-bits of the $e$'s, all located to the left of~$d$.
  The resulting state is thus identical to $w$ except that all MRU-bits to the left of $g_0$ have been flipped to~$1$;
  thus after accessing $g_0$, the state is identical to the initial state except that it ends with $g_0^{1-s} g_1^s$ instead of $g_0^s g_1^{1-s}$.
  
  Now consider the latter case: after the sequence of $a$'s there is at least one position of the form $a^1_{i,\beta}$ to the left of $d$.
  This position cannot be overwritten by the $e$'s.
  After the path through $\phi_{i,b}$, the state is thus of the form  $x_1^1, \dots, x_r^1, d^1, x_{r+1}^1, \dots, x_{2r}^1, g_0^0, g_1^1$, and one of the $x_i$ for $1 \leq i \leq r$ is an~$a$.
  The access to $g_0$ yields $x_1^0, \dots, x_r^0, d^0, x_{r+1}^0, \dots, x_{2r}^0, g_0^{1-s}, g_1^s$.
  This state is well-phased but not well-formed. 
\end{proof}

\begin{lemma}
  Executing a path through $\psi_{i,b} g_s$ over an NMRU state well-phased at step $s$ always leads to a state well-phased at step $1-s$.
  Furthermore that state
  \begin{itemize}
  \item either is not well-formed at step $1-s$
  \item or is identical to the initial state except for the $g_0$ and $g_1$ blocks, and, possibly, the $a_{i,\beta_i}$ block replaced by~$a_{i,b}$.
  \end{itemize}  
\end{lemma}

\begin{proof}
  Again, the initial state is $x_1^0, \dots, x_r^0, d^0, x_{r+1}^0, \dots, x_{2r}^0, g_0^0, g_1^1$ where the $(x_i)_{1 \leq i \leq 2r}$ are a permutation of
  $\{ e_i \mid 1 \leq i \leq r \} \cup
  \{ a_{i,\beta_i} \mid 1 \leq i \leq r \}$ for some sequence of Booleans
  $(\beta_i)_{1 \leq i \leq r}$.

  Consider a path through $\psi_{i,b}$: it consists of $d$, followed by a sequence of $r-1$ $a$'s, then $r$ $e$'s, then $a_{i,b}$.
  Each of these accesses either freshens, either overwrites, one of the $x$ positions.
  After the sequence of $a$'s, there are either no $a$'s to the left of $d$, or at least one.
  The former case is possible only if all these $a$'s are hits, freshening positions to the right of $d$.
  Then the sequence of $e$'s freshens the $e$'s to the left of~$d$.
  There is one remaining $x$ position with a zero MRU-bit: it is to the right of $d$ and carries a block $a_{i,\beta_i}$.
  This block is then updated or freshened by the $a_{i,b}$ access.
  Then the access to $g_0$ flips all MRU-bits to $0$ except the one for $g_0$, which is flipped to $1$.
  Since all of the accesses before the $a_{i,b}$ access were hits, the permutation of the positions has not changed: the state is the same as the initial state except that $a_{i,\beta_i}^0$ is replaced by $a_{i,b}^0$ and $g_0^s g_1^{1-s}$ is replaced by $g_0^{1-s} g_1^s$.

  Now consider the latter case: after the sequence of $a$'s there is at least one position of the form $a^1_{i,\beta}$ to the left of $d$.
  Then, as in the proof of the previous lemma, there is still $a^0_{i,\beta}$ to the left of $d$ at the end of the path through $\psi_{i,b} g_s$.
  Thus the final state cannot be well-formed.
\end{proof}

\begin{corollary}
  Assume starting in a well-formed NMRU state, corresponding to Boolean state $\sigma$, then any path through the gadget encoding an assignment or a guard
  \begin{itemize}
  \item either leads to a well-formed NMRU state, corresponding to the state $\sigma'$ obtained by executing the assignment, or $\sigma'=\sigma$ for a valid guard;
  \item or leads to a well-phased but not well-formed state.
  \end{itemize}
\end{corollary}

\begin{lemma}
  Executing the sequence $a_{1,\false} \dots a_{\nRegs,\false} c_1 \dots c_\nRegs$ from $F_a$ to $F_h$ over a well-formed NMRU state corresponding to a zero Boolean state leads to a state containing~$d$ --- more specifically, a state of the form $c_1^1,\dots,c_{\nRegs}^1,d^0,a_{\pi(1),\false}^1,\dots,a_{\pi(\nRegs),\false}^1 g_0^0 g_1^1$ where $\pi$ is a permutation.
\end{lemma}

\begin{proof}
  The $a_{1,\false} \dots a_{\nRegs,\false}$ just freshen the corresponding blocks (MRU-bit set to $1$), and then the $c_1 \dots c_\nRegs$ overwrite the $e$'s.
\end{proof}

\begin{lemma}
  Executing the sequence $a_{1,\false} \dots a_{\nRegs,\false} c_1 \dots c_\nRegs$ from $F_a$ to $F_h$ over a well-phased but not well-formed NMRU state leads to a state not containing~$d$ --- where the $2\nRegs$ first MRU bits are set to $1$, the next one to $0$, and then $g_0^0 g_1^1$.
\end{lemma}

\begin{proof}
  The well-phased but not well-formed NMRU state contains at least one $b$ to the right of $d$.
  When applying $a_{1,\false} \dots a_{\nRegs,\false}$, at least one of the $a$'s must thus freshen or replace a block to the left of $d$.
  Then when applying $c_1 \dots c_\nRegs$, $d$ gets erased.
\end{proof}

\begin{corollary}
  There is an execution sequence from $I_f$, with empty cache, to $F_f$, such that the final cache contains $d$ if and only if there is an execution trace from $I_r$ to~$F_r$.
\end{corollary}

\begin{theorem}
  The exist-hit problem for NMRU caches is NP-complete for acyclic graphs and PSPACE-complete for general graphs.
\end{theorem}

\subsection{Reduction to Exist-Miss}
\begin{definition}\label{def:NMRU_reduction_EM}
  We modify the reduction of \autoref{def:NMRU_reduction} as follows.
  Between $F_h$ and $F_f$ we insert the sequence $d g_0 c_1 \dots c_\nRegs a_{1,\true}$, as the third part of the epilogue.
\end{definition}

\begin{lemma}
  Executing $d g_0 c_1 \dots c_\nRegs a_{1,\true}$ over a state of the form $c_1^1,\dots,c_{\nRegs}^1,d^0,a_{\pi(1),\false}^1,\dots,a_{\pi(\nRegs),\false}^1 g_0^0 g_1^1$, where $\pi$ is a permutation, leads to a state without~$d$.
\end{lemma}

\begin{proof}
  $d$ gets freshened,
  then $g_0$ is the sole block with a zero MRU-bit.
  Thus, when it is freshened, all other MRU bits are set to zero.
  Then $c_1 \dots c_\nRegs$ freshen the first $\nRegs$ blocks,
  and $a_{1,\true}$ erases~$d$.
\end{proof}

\begin{lemma}
  Executing $d g_0 c_1 \dots c_\nRegs a_{1,\true}$ over a state not containing $d$, where the $2\nRegs$ first MRU bits are set to $1$, the next one to $0$, and then $g_0^0 g_1^1$ leads to a state containing~$d$.
\end{lemma}

\begin{proof}
  $d$ overwrites the $2\nRegs+1$-th position,
  then $g_0$ is the sole block with a zero MRU-bit.
  Thus, when it is freshened, all other MRU bits are set to zero.
  Then possibly some blocks get overwritten among the $\nRegs+1$ first blocks,
  and $d$ is still in the cache.
\end{proof}

\begin{corollary}
  There is an execution sequence from $I_f$, with empty cache, to $F_f$, such that the final cache does not contain $d$ if and only if there is an execution of the Boolean register machine from $I_r$ to~$F_r$.
\end{corollary}

\begin{theorem}
  The exist-miss problem for NMRU caches is NP-complete for acyclic graphs and PSPACE-complete for general graphs.
\end{theorem}

We have no results for reductions to cache analysis problems with arbitrary starting state.
The proof method that we used for FIFO and PLRU --- prepend a sufficiently long sequence of accesses that will bring the cache to a sufficiently known state --- does not seem to easily carry over to NMRU caches.
Even though it is known that $2\nWays-2$ pairwise distinct accesses are sufficient to remove all previous content from an NMRU cache~\cite[Th.~4]{DBLP:journals/rts/ReinekeGBW07}, it can be shown that there is no sequence guaranteed to yield a completely known cache state~\cite[Th.~5]{DBLP:journals/rts/ReinekeGBW07}.


\section{Conclusion and future work}
We have shown complexity-theoretical properties that separate LRU from other policies such as PLRU, NMRU, FIFO, pseudo-RR and that give a more rigorous meaning to the intuition that static analysis for LRU is ``easier'' than for other policies, due to its forgetful tendencies.
Reachability problems for PLRU, NMRU, FIFO, pseudo-RR caches in fact are of the same complexity class (PSPACE-complete) as generic reachability problems with arbitrary transitions over bits~\cite{DBLP:conf/stacs/FeigenbaumKVV98}, which again justifies why they are hard.

Our assumption that all paths through the control-flow graph are feasible is not realistic. As explained in the introduction, this assumption is made because arbitrary arithmetic and tests on lead to all cache analysis problems being equivalent to Turing's halting problem.
Another option is restricting the program to a finite number of bits and a transition relation represented by a Boolean circuit, formula, or OBDD, linking the previous values of the bits to the next values;
but as we have explained, this would not help distinguishing between policies since all problems then become PSPACE-complete~\cite{DBLP:conf/stacs/FeigenbaumKVV98}.
We thus would need some kind of weaker execution model, but it is difficult to find one that would still be realistic enough to have an interest.

One possibility, suggested by a reviewer, is to add a call stack. In Section~\ref{sec:fixed_associativity}, we used the fact that the combination of a finite control automaton and the exact state of the cache is just a bigger finite automaton, thus exist-miss and exist-hit can be decided by explicit-state model checking.
Adding a call stack, for modeling procedure calls, would turn the problem into model-checking reachability on a stack automaton, in other words a pushdown system.
Despite these systems having infinite state, reachability is decidable~\cite{DBLP:conf/concur/BouajjaniEM97,DBLP:journals/entcs/FinkelWW97}.
Investigating the complexity of the combination of the call stack with our cache models is left to future work.

Another direction is to distinguish the analysis difficulty for different policies for a fixed associativity. As explained in Section~\ref{sec:fixed_associativity}, all problems become polynomial with respect to the program size; but one could still want to distinguish asymptotic growths.
However, as we have shown, too strong results in this area would entail answers to very hard conjectures in complexity theory.

It is difficult to draw practical implications from worst-case complexity-theoretic asymptotic results. We however think that our results are yet another indication, in addition to the existence of efficient and precise analyses for LRU~\cite{Ferdinand99,Touzeau_et_al_CAV2017,Touzeau_et_al_POPL19} and their lack for other policies, that the LRU policy is to be preferred for ease of analysis and thus for hard real time critical applications where a worst case execution time bound must be established.%
\footnote{We have anecdotal evidence that certain designers of safety critical systems lock 6 ways out of 8 PLRU ways in the MPC755 processor~\cite[Table C.3]{MPC750_Manual} so that the remaining 2 ways are equivalent to a 2-way LRU cache, amenable to analysis. This illustrates the cost of using a policy that performs well ``on average'' but that is not easy to statically predict.}

\printbibliography

@string{LNCS = {Lecture Notes in Computer Science}}

@String{Springer = {Springer Verlag}}

@string{CAV = {Computer-aided verification (CAV)}}

@TechReport{MRU_patent,
  author = 	 {Malamy, Adam and 
                  Patel, Rajiv N. and
                  Hayes, Norman M.},
  title = 	 {Methods and apparatus for implementing a pseudo-LRU cache memory replacement scheme with a locking feature},
  institution =  {US Patent Office},
  year = 	 1994,
  type = 	 {US patent},
  number = 	 {5,353,425},
  month = 	 oct,
  url = {https://patents.google.com/patent/US5353425}
}

@inproceedings{Al-Zoubi:2004:PEC:986537.986601,
 author = {Al-Zoubi, Hussein and Milenkovic, Aleksandar and Milenkovic, Milena},
 title = {Performance Evaluation of Cache Replacement Policies for the SPEC CPU2000 Benchmark Suite},
 booktitle = {Proceedings of the 42Nd Annual Southeast Regional Conference},
 series = {ACM-SE 42},
 year = {2004},
 isbn = {1-58113-870-9},
 location = {Huntsville, Alabama},
 pages = {267--272},
 znumpages = {6},
 zurl = {http://doi.acm.org/10.1145/986537.986601},
 doi = {10.1145/986537.986601},
 acmid = {986601},
 publisher = {ACM},
 address = {New York, NY, USA},
}

@InProceedings{Touzeau_et_al_CAV2017,
  author = 	 {Valentin Touzeau and Claire Ma{\"i}za and David Monniaux and Jan Reineke},
  title = 	 {Ascertaining Uncertainty for Efficient Exact Cache Analysis},
  booktitle = CAV,
  year = 	 2017,
  volume = 10427,
  pages = {22--17},
  editor = 	 {Viktor Kuncak and Rupak Majumdar},
  publisher = Springer,
  address = {Cham},
  eprint = {1709.10008},
  eprinttype = {arXiv},
  entrysubtype = "intc",
  doi = {10.1007/3-540-63141-0_10}
}

@article{Ferdinand99,
 author = {Ferdinand, Christian and Wilhelm, Reinhard},
 title = {Efficient and Precise Cache Behavior Prediction for Real-Time Systems},
 journal = {Real-Time Systems},
 issue_date = {Nov. 1999},
 volume = {17},
 number = {2--3},
 month = dec,
 year = {1999},
 issn = {0922-6443},
 pages = {131--181},
 znumpages = {51},
 doi = {10.1023/A:1008186323068},
 acmid = {338858},
 fpublisher = {Kluwer Academic Publishers},
 publisher = {Kluwer},
 address = {Norwell, MA, USA},
}

@PhdThesis{Reineke_PhD,
  author = 	 {Jan Reineke},
  title = 	 {Caches in WCET analysis: predictability, competitiveness, sensitivity},
  school = 	 {Universit{\"a}t des Saarlandes},
  year = 	 {2008},
  url = {http://www.rw.cdl.uni-saarland.de/~reineke/publications/DissertationCachesInWCETAnalysis.pdf}
}

@Manual{MPC7450_Manual,
  title = 	 {MPC7450 RISC Microprocessor Family Reference Manual},
  organization = {NXP / Freescale Semiconductor},
  edition = 	 5,
  month = 	 jan,
  year = 	 2005,
  url = {https://www.nxp.com/docs/en/reference-manual/MPC7450UM.pdf}
}

@Manual{MPC750_Manual,
  title = 	 {MPC750 RISC Microprocessor Family Reference Manual},
  organization = {NXP / Freescale Semiconductor},
  edition = 	 1,
  month = dec,
  year = 	 2001,
  url = {https://www.nxp.com/docs/en/reference-manual/MPC750UM.pdf}
}

@InProceedings{berg:OASIcs:2006:672,
  author ={Christoph Berg},
  title ={{PLRU} Cache Domino Effects},
  booktitle ={6th International Workshop on Worst-Case Execution Time Analysis (WCET'06)},
  series ={OpenAccess Series in Informatics (OASIcs)},
  ISBN ={978-3-939897-03-3},
  ISSN ={2190-6807},
  year ={2006},
  volume ={4},
  editor ={Frank Mueller},
  publisher ={Schloss Dagstuhl--Leibniz-Zentrum fuer Informatik},
  address ={Dagstuhl, Germany},
  zURL ={http://drops.dagstuhl.de/opus/volltexte/2006/672},
  URN ={urn:nbn:de:0030-drops-6723},
  doi ={10.4230/OASIcs.WCET.2006.672},
  pages = {69--71}
}

@Manual{i486_data_sheet,
  title = 	 {i486 Microprocessor data sheet},
  organization = {Intel},
  year = 	 1989,
  url = {https://archive.org/stream/bitsavers_intel80486ataSheetApr89_12763574/i486_Microprocessor_Data_Sheet_Apr89_djvu.txt}
}

@article{DBLP:journals/pieee/HeckmannLTW03,
  author    = {Reinhold Heckmann and
               Marc Langenbach and
               Stephan Thesing and
               Reinhard Wilhelm},
  title     = {The influence of processor architecture on the design and the results
               of {WCET} tools},
  journal   = {Proceedings of the {IEEE}},
  volume    = {91},
  number    = {7},
  pages     = {1038--1054},
  year      = {2003},
  bibsource = {dblp computer science bibliography, https://dblp.org},
  doi = {10.1109/JPROC.2003.814618}
}

@article{Touzeau_et_al_POPL19,
  author    = {Valentin Touzeau and
               Claire Ma{\"{\i}}za and
               David Monniaux and
               Jan Reineke},
  title     = {Fast and exact analysis for {LRU} caches},
  journal   = {Proceedings of the ACM on Programming Languages (PACMPL)},
  volume    = {3},
  number    = {{POPL}},
  pages     = {54:1--54:29},
  year      = {2019},
  zurl       = {https://doi.org/10.1145/3290367},
  doi       = {10.1145/3290367},
  bibsource = {dblp computer science bibliography, https://dblp.org},
  eprint = {1811.01670},
  eprinttype = {arXiv}
}

@article{DBLP:journals/rts/ReinekeGBW07,
  author    = {Jan Reineke and
               Daniel Grund and
               Christoph Berg and
               Reinhard Wilhelm},
  title     = {Timing predictability of cache replacement policies},
  journal   = {Real-Time Systems},
  volume    = {37},
  number    = {2},
  pages     = {99--122},
  year      = {2007},
  zurl       = {https://doi.org/10.1007/s11241-007-9032-3},
  doi       = {10.1007/s11241-007-9032-3},
  url = {http://www.rw.cdl.uni-saarland.de/~grund/papers/rts07-predictability.pdf},
  bibsource = {dblp computer science bibliography, https://dblp.org}
}

@inproceedings{DBLP:conf/stacs/FeigenbaumKVV98,
  author    = {Joan Feigenbaum and
               Sampath Kannan and
               Moshe Y. Vardi and
               Mahesh Viswanathan},
  editor    = {Michel Morvan and
               Christoph Meinel and
               Daniel Krob},
  title     = {Complexity of Problems on Graphs Represented as OBDDs (Extended Abstract)},
  booktitle = {Symposium on Theoretical Aspects of Computer (STACS)},
  series    = LNCS,
  volume    = {1373},
  pages     = {216--226},
  publisher = Springer,
  address   = {Berlin, Heidelberg},
  year      = {1998},
  doi       = {10.1007/BFb0028563},
  bibsource = {dblp computer science bibliography, https://dblp.org}
}

@inproceedings{DBLP:conf/concur/BouajjaniEM97,
  author    = {Ahmed Bouajjani and
               Javier Esparza and
               Oded Maler},
  editor    = {Antoni W. Mazurkiewicz and
               J{\'{o}}zef Winkowski},
  title     = {Reachability Analysis of Pushdown Automata: Application to Model-Checking},
  booktitle = {Concurrency Theory (CONCUR)},
  series    = LNCS,
  volume    = {1243},
  pages     = {135--150},
  publisher = Springer,
  address   = {Berlin, Heidelberg},
  year      = {1997},
  doi       = {10.1007/3-540-63141-0_10},
  bibsource = {dblp computer science bibliography, https://dblp.org}
}

@article{DBLP:journals/entcs/FinkelWW97,
  author    = {Alain Finkel and
               Bernard Willems and
               Pierre Wolper},
  title     = {A direct symbolic approach to model checking pushdown systems},
  journal   = {Electr. Notes Theor. Comput. Sci.},
  volume    = {9},
  pages     = {27--37},
  year      = {1997},
  zurl       = {https://doi.org/10.1016/S1571-0661(05)80426-8},
  doi       = {10.1016/S1571-0661(05)80426-8},
  bibsource = {dblp computer science bibliography, https://dblp.org}
}
\end{document}